\documentclass[a4paper,USenglish,cleveref, autoref, thm-restate, numberwithinsect,notab]{lipics-v2021}

\hideLIPIcs  

\title{An ETH-Tight FPT Algorithm for Rejection-Proof Set Packing with Applications to Kidney Exchange}

\titlerunning{An ETH-Tight FPT Algorithm for Rejection-Proof Set Packing}

\author{Bart M.\,P. Jansen}
{Eindhoven University of Technology, The Netherlands}
{b.m.p.jansen@tue.nl}
{https://orcid.org/0000-0001-8204-1268}{Supported by the Dutch Research Council (NWO) through Gravitation-grant NETWORKS-024.002.003.}

\author{Jeroen S.\,K. Lamme}
{Eindhoven University of Technology, The Netherlands}
{j.s.k.lamme@tue.nl}
{https://orcid.org/0009-0005-8901-2271}
{Supported by the Dutch Research Council (NWO) through Gravitation-grant NETWORKS-024.002.003.}

\author{Ruben F.\,A. Verhaegh}
{Eindhoven University of Technology, The Netherlands}
{r.f.a.verhaegh@tue.nl}
{https://orcid.org/0009-0008-8568-104X}
{}

\authorrunning{B.M.P.~Jansen, J.S.K.~Lamme, and R.F.A.~Verhaegh}

\Copyright{B.M.P.~Jansen, J.S.K.~Lamme, and R.F.A.~Verhaegh}

\ccsdesc[500]{Theory of computation~Parameterized complexity and exact algorithms}
\ccsdesc[300]{Applied computing~Health informatics}
\ccsdesc[300]{Theory of computation~Algorithmic game theory}
\ccsdesc[100]{Mathematics of computing~Combinatorial algorithms}

\keywords{Parameterized complexity, Multi-agent kidney exchange, Kernelization, Set packing}

\category{} 

\relatedversion{} 

 \nolinenumbers 

\EventEditors{John Q. Open and Joan R. Access}
\EventNoEds{2}
\EventLongTitle{42nd Conference on Very Important Topics (CVIT 2016)}
\EventShortTitle{CVIT 2016}
\EventAcronym{CVIT}
\EventYear{2016}
\EventDate{December 24--27, 2016}
\EventLocation{Little Whinging, United Kingdom}
\EventLogo{}
\SeriesVolume{42}
\ArticleNo{23}
\usepackage{lineno}
\newsavebox{\socgnlbox}%

\newcommand{\problem}[3]{
  \begin{nolinenumbers}
    \vspace{1mm}
    \noindent\fbox{
        \begin{minipage}{0.96\textwidth}
        \begin{tabularx}{\textwidth}{@{\hspace{\parindent}}l X}
            \multicolumn{2}{@{\hspace{\parindent}}l}{\textsc{#1}} \\
            \textbf{Input:} & #2 \\
            \textbf{Task:} & #3 \\
        \end{tabularx}
        \end{minipage}
    }
  \vspace{1mm}
  \end{nolinenumbers}
}

\newtheorem{rrule}{Reduction rule}

\newtheorem{fact}[theorem]{Fact}
\Crefname{claim}{Claim}{Claims}

\newcommand{\Oh}{\mathcal{O}}

\newcommand{\C}{\mathcal{C}}
\renewcommand{\S}{\mathcal{S}}
\newcommand{\X}{\mathcal{X}}

\newcommand{\Crej}{\C_\mathrm{rej}}
\newcommand{\Cint}{\C_\mathrm{int}}
\newcommand{\Xrej}{\X_\mathrm{rej}}
\newcommand{\Xint}{\X_\mathrm{int}}

\begin{document}

\maketitle
\begin{abstract}
We study the parameterized complexity of a recently introduced multi-agent variant of the \textsc{Kidney Exchange} problem. Given a directed graph~$G$ and integers~$d$ and $k$, the standard problem asks whether~$G$ contains a packing of vertex-disjoint cycles, each of length~$\leq d$, covering at least~$k$ vertices in total. In the multi-agent setting we consider, the vertex set is partitioned over several agents who reject a cycle packing as solution if it can be modified into an alternative packing that covers more of their own vertices. A cycle packing is called \emph{rejection-proof} if no agent rejects it and the problem asks whether such a packing exists that covers at least~$k$ vertices. 

We exploit the sunflower lemma on a set packing formulation of the problem to give a kernel for this~$\Sigma_2^P$-complete problem that is polynomial in $k$ for all constant values of $d$. We also provide a $2^{\Oh(k \log k)} + n^{\Oh(1)}$ algorithm based on it and show that this FPT algorithm is asymptotically optimal under the ETH. Further, we generalize the problem by including an additional positive integer~$c$ in the input that naturally captures how much agents can modify a given cycle packing to reject it. For every constant~$c$, the resulting problem simplifies from being~$\Sigma_2^P$-complete to NP-complete. The super-exponential lower bound already holds for~$c=2$, though. We present an ad-hoc single-exponential algorithm for~$c = 1$. 
These results reveal an interesting discrepancy between the classical and parameterized complexity of the problem and give a good view of what makes it hard.
\end{abstract}
\newpage
\section{Introduction}
\label{sec:introduction}
The goal of this paper is to analyze the parameterized complexity of algorithmic problems that arise from a recently introduced multi-party variant of kidney exchange problems. Before describing our results, we first present some relevant context.

For patients with severely reduced kidney function, transplantation is often the preferred treatment~\cite{Burra07}. As it is possible to live with only one healthy kidney, live kidney donation is employed to reduce waiting times for sparsely available donor organs. In this setting, a willing donor offers one of their kidneys for transplantation to a patient. The patient and donor must be medically compatible for a successful transplantation. When the patient and donor are not compatible, they can opt to participate in a kidney exchange program (KEP). 

Such programs aim to link multiple patient-donor pairs so that a donor from one pair is compatible with a patient from another pair. Their goal is to group sets of patient-donor pairs, such that within each group, a cyclic order of compatibility exists. Performing this cycle of transplants means that donors only undergo surgery if their corresponding patient receives a kidney in that exchange. Individual programs each have their own additional restrictions~\cite{Biro21}. It is common to restrict the length of transplantation cycles to some constant~$d$, due to practical considerations which include doing the associated surgeries simultaneously.

We study these KEPs from an algorithmic perspective. Given a set of participating patient-donor pairs and the medical compatibilities between them, the algorithmic task becomes to find a combination of transplantation cycles that maximizes the total number of patients receiving a healthy kidney. We can model each patient-donor pair as a vertex and the compatibility of the donor of a pair~$a$ to the patient of a pair~$b$ as a directed edge from~$a$ to~$b$. We call a cycle of length at most~$d$ a~\emph{$d$-cycle} and use this notation to obtain the following algorithmic graph problem that describes KEPs in one of their most basic forms.

\problem{Kidney Exchange with $d$-cycles (KE-$d$)}
{A directed graph $G$ and an integer $k$.}
{Determine whether there is a set $\C$ of pairwise vertex-disjoint $d$-cycles such that $\C$ covers at least $k$ vertices.}

This problem has already been studied from the perspective of parameterized algorithms in the literature, with fixed-parameter tractable (FPT) algorithms w.r.t.~the input parameter~$k$ and the number of neighborhood equivalence classes as examples~\cite{Hebert-JohnsonL24, MaitiD22}.

\subparagraph{Multi-agent KEPs} Extensions of the standard kidney exchange problem have also been studied. In this work, we consider one such extension which we now motivate and introduce. It is known that the fraction of patients to receive a healthy kidney can increase significantly with an increase in the total number of participating pairs~\cite{Agarwal19}. This motivates collaborations between multiple programs, from now on referred to as agents. Such collaborations, however, can be hindered by a mismatch between the individual interests of agents, as these need not align with the collective goal of helping as many patients as possible overall. In particular, a socially optimal solution---in which as many patients as possible are treated overall---is not guaranteed to help as many people from each individual agent as that agent could have achieved on their own. This may stop agents from participating in collaborative programs or, when they do participate, move them to withhold information from the collective.

Several ways of incorporating fairness into multi-agent KEPs have been studied~\cite{Agarwal19,Ashlagi15,Biro20,Carvalho23, Carvalho16, Hajaj15, Klimentova21}. 
Recently, Blom, Smeulders, and Spieksma~\cite{Blom24} also introduced a framework to capture fairness via a game-theoretical concept of \textit{rejections}. They translate this into an algorithmic task that asks for a solution that is not rejected by any agent. We study the resulting algorithmic task. To introduce it, we first define the concept of rejections.

Recall the original KE-$d$ problem on directed graphs. To encode multiple agents into the input, we partition the vertex set of the input graph over the agents. In the obtained graph, we refer to a cycle that only contains vertices from a single agent $i$ as an \emph{$i$-internal} cycle. Like in the original~KE-$d$ problem, a solution consists of a packing $\C$ of pairwise vertex-disjoint $d$-cycles. A given such packing $\C$ will be rejected by an agent under the conditions listed in the definition below. See also \autoref{fig:rejection-example} for an example, adapted from a figure by Blom et al.

\begin{figure}[t]
    \centering
        \begin{minipage}[t]{.445\linewidth}
                \centering                \includegraphics[width=.8\linewidth]{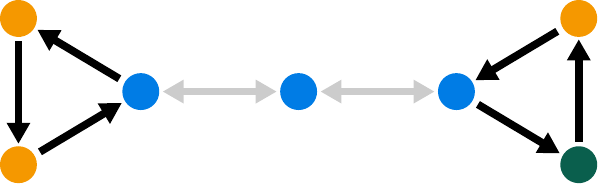}
                \label{fig:rejection-example-1}
        \end{minipage}
        \hfill
        \begin{minipage}[t]{.445\linewidth}
                \centering                \includegraphics[width=.8\linewidth]{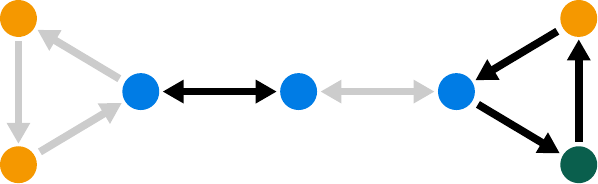}
                \label{fig:rejection-example-2}
        \end{minipage}
        \caption{A socially optimal packing of cycles is shown on the left and covers six vertices. However, the blue agent rejects this solution as they could modify it to cover more blue vertices as shown on the right. They can do this by rejecting the left-most $3$-cycle and including an internal $2$-cycle.}
        \label{fig:rejection-example}
    \end{figure}
\begin{definition} \label{def:rejection}
    Let $G$ be a graph, where $V(G) = V_1 \cup \ldots \cup V_p$ is partitioned over~$p$ different agents. Let $d$ be an upper bound on the length of allowed transplantation cycles. Let~$\C$ be a set of pairwise vertex-disjoint $d$-cycles in $G$. If there is an agent $i \in [p]$ for which there exist
    \begin{itemize}
        \item a subset $\Crej \subseteq \C$, and
        \item a set $\Cint$ of $i$-internal $d$-cycles, such that $\C':=(\C \setminus \Crej) \cup \Cint$ is a set of cycles that are still pairwise vertex-disjoint while $\C'$ covers more vertices from $V_i$ than $\C$ does,
    \end{itemize}
     then, agent $i$ is said to \emph{reject} the packing $\C$. Equivalently, we may say that agent $i$ specifically rejects the cycles in~$\Crej$ and we refer to $\C'$ as an \emph{alternative packing} by which agent $i$ rejects.
\end{definition}
This definition captures that the rejecting agent~$i$ could modify a given cycle packing $\C$ into another packing $\C'$ that benefits more of their own patients by only making changes pertaining to patients within their control. We call a packing that is not rejected by any agent \textit{rejection-proof}. Then, a guarantee to always find such a solution also serves as a guarantee for agents that participation in the joint program is never detrimental to their personal interests. This led Blom et al.~\cite{Blom24} to introduce the following problem, whose parameterized complexity we investigate in this work.

\problem{Rejection-Proof Kidney Exchange with $d$-cycles (RPKE-$d$)}
{A directed graph $G$ where $V(G) = V_1 \cup \ldots \cup V_p$ is partitioned among $p$ agents, and an integer $k$.}
{Determine whether there is a \emph{rejection-proof} set $\C$ of pairwise vertex-disjoint $d$-cycles such that $\C$ covers at least $k$ vertices.}
For ease of notation, we refer to any packing of pairwise vertex-disjoint $d$-cycles that covers at least $k$ vertices in an RPKE-$d$ instance as a \emph{candidate solution}. Hence, a packing of $d$-cycles forms a solution to an instance if it is both a candidate solution and rejection-proof.

The above problem takes a leap in complexity compared to the single-agent problem KE-$d$. Whereas KE-$d$ is known to be NP-complete for all finite~$d \geq 3$~\cite{Abraham07}, Blom et al. have shown RPKE-$d$ to be $\Sigma_2^P$-complete for~$d \geq 3$, even when there are only two participating agents~\cite[Theorem 1]{Blom24}. In fact, it can be shown that checking whether a given agent rejects a given candidate solution is already NP-hard since this is at least as hard as checking whether a given packing is optimal in the NP-hard setting of only one agent. Partly motivated by this observation, we also consider a more restricted definition of rejections.

For some integer $c \geq 0$, we say that an agent \emph{$c$-rejects} a packing $\C$ if the rejection criterion from \autoref{def:rejection} can be satisfied with~$|\Crej| \leq c$. We call a packing that no agent $c$-rejects \emph{$c$-rejection-proof}.
Now, for each pair of positive integers~$c$ and~$d$ we define the problem \textsc{$c$-Rejection-Proof Kidney Exchange with $d$-cycles} ($c$-RPKE-$d$) as the variant of the above RPKE-$d$ problem in which we aim to determine whether the input contains a candidate solution that is $c$-rejection-proof.

The integer $c$ in this definition indicates some level of complexity that a rejection can have. Additionally, it provides a way to interpolate between the problems KE-$d$ ($c=0$) and RPKE-$d$ ($c=\infty$). More importantly, for any constant~$c$, we observe that the task of determining whether a given agent $c$-rejects a given packing $\C$ of $d$-cycles is polynomial-time solvable when $d$ is also constant. This is witnessed by the fact that, if an agent $c$-rejects the packing $\C$, an alternative packing can always be obtained by rejecting a constant number of cycles and including a constant number of internal cycles. It follows that $c$-RPKE-$d$ is contained in NP when $c$ and $d$ are constants, as opposed to the $\Sigma_2^P$-complete RPKE-$d$ problem that does not limit the number of cycles that agents are allowed to reject.

\subparagraph{Our contribution} We study the complexity of RPKE-$d$ and $c$-RPKE-$d$ parameterized by~$k$, the number of vertices to be covered, and we present both positive and negative results. We prove our positive results on related set packing problems that generalize the two graph problems. This allows us to use some of the tools developed for such problems. Prior to our work, the only positive result known for RPKE-$d$ was a polynomial-time algorithm for~$d=2$~\cite{Carvalho23}.

We give an algorithm that, for constant values of~$d$, solves RPKE-$d$ in $2^{\Oh(k \log k)} + n^{\Oh(1)}$ time on $n$-vertex graphs. It consists of two parts, the first of which is a kernelization algorithm. It takes an arbitrary instance as input and outputs an equivalent instance of size $\Oh\left( (k^{d+1})^d \right)$ in polynomial time, via a careful agent-aware application of the sunflower lemma~\cite{Erdos60} to eliminate transplantation cycles. 
In set packing problems, removing sets during preprocessing typically comes at the risk of removing sets that are required to make an optimal solution, thereby possibly transforming YES-instances into NO-instances. The sunflower lemma provides a tool to circumvent this. Our rejection-proof problem settings introduce an additional type of behavior that needs care: transplantation cycles that do not belong to optimal solutions may still play a role in rejections. Removing the wrong cycles may cause an agent not to reject a certain candidate solution that it otherwise would have. 
We present a two-step reduction approach guaranteed to preserve the answer. The second part of the algorithm consists of brute-force on the reduced instance. 

We give a reduction from \textsc{Subgraph Isomorphism} to show that, under the Exponential Time Hypothesis (ETH), the algorithm is asymptotically optimal in terms of $n$ and $k$. We even show the stronger bound (independent of $k$) that a $2^{o(n \log n)}$ time algorithm for RPKE-$d$ or $c$-RPKE-$d$ would violate the ETH for any $c \geq 2$ and $d \geq 3$. We exploit the mechanism of rejections to facilitate our reduction from the \textsc{Subgraph Isomorphism} problem, for which such a superexponential ETH-based lower bound is known~\cite{CyganFGKMPS16}. This gives the rejection-proof kidney exchange problems a similar complexity status to other problems in the literature, such as \textsc{Metric Dimension} and \textsc{Geodetic Set} parameterized by vertex cover number, for which a brute-force algorithm on a polynomial kernel is also ETH-optimal~\cite{FoucaudGK0IST25}. 

To complement the ETH-hardness of $c$-RPKE-$d$ with $c \geq 2$, we give a $3^n \cdot n^{\Oh(1)}$ time algorithm for the problem when $c=1$. This allows us to see surprising jumps in the algorithmic complexity of $c$-RPKE-$d$ for varying values of $c$. While no big difference in complexity is obtained between $c=1$ and $c=0$ (in which case the problem can be solved in~$2^n \cdot n^{\Oh(1)}$ time with a standard application of dynamic programming over subsets~\cite[Chapter 6]{CyganFGKMPS16}), our results show that the problem does become computationally harder from a parameterized perspective when $c=2$ (assuming the ETH). We also see that no more jumps in parameterized complexity appear for larger values of $c$. The surprising aspect here is that this one discrete jump in the parameterized complexity of $c$-RPKE-$d$ does not coincide with the jump in the classical complexity of the problem between being contained in NP vs.~not being contained in NP: we have seen that, for constant $d\ge 3$, the $c$-RPKE-$d$ problem is contained in NP for \emph{every} constant value of $c$, while the RPKE-$d$ problem is $\Sigma_2^P$-complete.

Another jump in complexity is observed when looking at the number of vertices that are \emph{not} covered by a solution. When asked to find a solution that covers all vertices in the input, the problem again collapses to the NP-complete packing problem KE-$d$. However, our ETH-hardness proof shows that finding a solution that covers all but one vertex is essentially as difficult as the general problem in terms of running time, assuming the ETH. Combined, our results yield a good understanding of what makes rejection-proof kidney exchange hard.

\subparagraph{Organization} The remainder of this work is organized as follows. We start with some preliminaries in \autoref{sec:preliminaries}. Next, we give a polynomial kernel and an FPT algorithm for a set packing generalization of RPKE-$d$ in \autoref{sec:kernel}, and we present a matching ETH-lower bound in \autoref{sec:eth-lower-bound}. In \autoref{sec:1-rejection-algo}, we give a single-exponential algorithm under $1$-rejections and we conclude in \autoref{sec:conclusion}. The proofs of statements marked $(\bigstar)$ are deferred to the appendix.

\section{Preliminaries}
\label{sec:preliminaries}
All graphs considered are directed, unless stated otherwise. We use standard notation for graphs and concepts from parameterized complexity; see the textbook by Cygan et al.~\cite{CyganFKLMPPS15}. We write $\langle v_1, \ldots, v_\ell \rangle$ to denote the directed cycle with edge set $\{(v_1, v_2), \ldots, (v_{\ell-1}, v_\ell), (v_\ell, v_1) \}$. We use the term \emph{packing} to refer to a collection of sets that satisfies certain pairwise disjointness conditions, such as collections of vertex-disjoint cycles in a graph or collections of pairwise-disjoint subsets of a universe. Since the disjointness conditions differ depending on the context, we always specify which condition holds for the packing at hand. 

A set system is a pair $(U, \S)$ where $U$ is a universe of elements and $\S \subseteq 2^U$ is a collection of subsets of this universe. For a subset $U' \subseteq U$, we use the notation $\S[U']$ to denote the collection of sets from $\S$ that contain only elements of $U'$. A \emph{hitting set} of $\S$ is a subset of~$U$ that intersects each set $S \in \S$. We use the following folklore lemma.

\begin{restatable}{lemma}{exactSetCoverLemma} $(\bigstar)$\label{lm:exact set cover}
    For a set system $(U, \S)$, it can be checked in $2^{|U|} \cdot (|U| + |\S|)^{\Oh(1)}$ time whether $\S$ contains a collection of pairwise disjoint sets that covers all elements of $U$.
\end{restatable}

Next, we present the sunflower lemma by Erdős and Rado~\cite{Erdos60} and a related definition of \emph{sunflowers} in set systems. In a set system~$(U, \S)$, a sunflower~$\mathcal{F}$ with~$z$ petals and \emph{core}~$Y$ is a collection of sets $S_1, \ldots, S_z \in \S$ such that $S_i \cap S_j = Y$ for all~$i \neq j$. We refer to a set~$S_i \setminus Y$ as the \emph{petal} corresponding to set~$S_i$ and we require all petals to be non-empty. We state a variant of it as proven in the textbook by Flum and Grohe~\cite[page~212]{FlumGrohe12}.

\begin{lemma}[Sunflower lemma]\label{lm:sunflower}
    Let $(U, \S)$ be a set system, such that each set in $\S$ has size at most $d$. If $|\S| > d \cdot d!(z-1)^d$, then the set system contains a sunflower with $z$ petals and such a sunflower can be computed in time polynomial in $|U|$, $|\S|$, and $z$.
\end{lemma}

The following lemma follows from the definitions of sunflowers and hitting sets. 

\begin{restatable}{lemma}{differentsolution}\label{lm: constructing a different solution}
    $(\bigstar)$ Let $(U, \S)$ be a set system with sets of size $\leq d$ for which there exists a hitting set of size $\leq h$, for some integers $d,h$. Let $\X \subseteq \S$ be a collection of disjoint sets and let $\mathcal{F} \subseteq \S$ be a sunflower of size $d(h-1)+2$ that contains one of the sets $X \in \X$. Then, there is a set $X' \in \mathcal{F} \setminus \{X\}$ such that $(\X \setminus {X}) \cup \{X'\}$ is a collection of pairwise disjoint sets.
\end{restatable} 

When analyzing the running time of an algorithm whose input is a graph, we denote the number of vertices and edges in the input graph by~$n$ and~$m$, respectively. For algorithms operating on a set system, we use~$n$ to denote the number of elements in the universe and~$m$ for the total number of sets in the system.

\subparagraph*{Rejection-proof set packing}
As mentioned, we present our positive results for problems that generalize the kidney exchange problems defined in the introduction. Intuitively, one may think of translating an RPKE-$d$ instance on a graph $G$ into this generalized setting by constructing a set system $(U, \S)$ with $U = V(G)$ and including a set in $\S$ for every vertex set that forms a $d$-cycle in $G$. Related definitions follow analogously, but we briefly reiterate and name them below.
To represent $p$ different agents in a set system $(U, \S)$, the universe $U = U_1 \cup \ldots \cup U_p$ is partitioned into $p$ sets. A set $S \in \S$ is called \emph{$i$-internal} if $S \subseteq U_i$.

\begin{definition} \label{def:set-rejection}
    Let $(U, \mathcal{S})$ be a set system such that the universe $U = U_1 \cup \ldots \cup U_p$ is partitioned over $p$ different agents. Let $\X \subseteq \S$ be a packing of pairwise disjoint sets. Suppose there is an agent $i \in [p]$ for which there exist:
    \begin{itemize}
        \item a subset $\Xrej \subseteq \X$, and
        \item a set $\Xint$ of $i$-internal sets such that $\X' := (\X \setminus \Xrej) \cup \Xint$ is a collection of sets that are still pairwise disjoint, while $\X'$ covers more elements from $U_i$ than $\X$ does,
    \end{itemize}
    then, agent $i$ is said to \emph{reject} the packing $\X$. Equivalently, we may say that agent $i$ specifically rejects the sets in~$\Xrej$ and we refer to $\X'$ as an alternative packing by which agent $i$ rejects.
\end{definition}

We call a packing that is not rejected by any agent \emph{rejection-proof}, and we use this notation to introduce the following generalized problem statement.

\problem{Rejection-Proof $d$-Set Packing (RP-$d$-SP)}
{A universe $U = U_1 \cup \ldots \cup U_p$ that is partitioned over $p$ agents, a collection $\mathcal{S} \subseteq 2^U$ of sets of size $\leq d$, and an integer $k$.}
{Determine whether there is a rejection-proof collection $\mathcal{X} \subseteq \mathcal{S}$ of disjoint sets that covers at least $k$ elements from~$U$.}

Like in the graph setting, we call any packing of pairwise disjoint sets that cover at least~$k$ vertices a \emph{candidate solution}. Further, for some integer $c \geq 0$, we say that an agent \emph{$c$-rejects} a given packing $\X$ if the rejection criterion from \autoref{def:set-rejection} can be satisfied with $|\Xrej| \leq c$. A packing that no agent $c$-rejects is called \emph{$c$-rejection-proof}. We define \textsc{$c$-Rejection-Proof $d$-Set Packing} ($c$-RP-$d$-SP) as a problem with the same input as RP-$d$-SP, but in which the goal is to determine whether the input contains a $c$-rejection-proof candidate solution.

\subparagraph*{Exponential Time Hypothesis (ETH)}
In \autoref{sec:eth-lower-bound}, we derive a hardness result that is conditional on the well-known ETH. This is a technical conjecture implying that the satisfiability of Boolean $3$-CNF formulas cannot be checked in subexponential time~\cite{ImpagliazzoP01}. Specifically, we use a known hardness result for the \textsc{Subgraph Isomorphism} problem conditional on the ETH. This problem takes two undirected graphs $G$ and $H$ as input and asks whether~$H$ is isomorphic to a subgraph of $G$. More formally, the problem asks whether there is an injective function $\varphi: V(H) \rightarrow V(G)$ such that $\{u,v\} \in E(H)$ implies $\{\varphi(u), \varphi(v)\} \in E(G)$. Such a function is called a subgraph isomorphism from~$H$ to~$G$. The following conditional lower bound on the computational complexity of this problem is known.

\begin{lemma}[{\cite{CyganFGKMPS16}}]
\label{lem:subgraph-isomorphism}
    Assuming the ETH, no algorithm exists that determines in $2^{o(n_G \log n_G)}$ time whether an undirected graph $H$ is isomorphic to a subgraph of $G$, where $n_G = |V(G)|$.
\end{lemma}
\section{A kernelization and FPT algorithm for rejection-proof packings} \label{sec:kernel}

In this section, we present a polynomial kernel and a derived FPT algorithm for RP-$d$-SP, which translates directly to RPKE-$d$ as explained in \autoref{sec:preliminaries}.
We start with two simple statements that are useful in proving the correctness of our kernel.
\begin{lemma}\label{lm: removing not used sets preserves solution}
    Let $I_1 = (U_1,\dots,U_p,\mathcal{S},k)$ be an RP-$d$-SP instance, let $\mathcal{X} \subseteq \mathcal{S}$ be a solution for it and let $X \in \mathcal{S}$. If $X \notin \mathcal{X}$, then $\mathcal{S}$ is a solution for the  instance $I_2 = (U_1,\dots,U_p,\mathcal{S} \setminus \{X\},k)$.
\end{lemma}
\begin{proof}
    Clearly, $\mathcal{X}$ is also a candidate solution in $I_2$. Now, suppose for contradiction that some agent $i$ rejects $\mathcal{X}$ by proposing an alternative packing $(\mathcal{X} \setminus \Xrej) \cup \Xint$. This alternative packing can also be constructed in $I_1$, contradicting that $\mathcal{X}$ is rejection-proof in $I_1$.
\end{proof}



\begin{observation}\label{ob: idential collections rejection iff}
    Let $S_i = \bigcup_{S \in \S[U_i]} S$ denote the elements which are contained in $i$-internal sets for some agent $i \in [p]$.
    If for collections $\X$ and $\X'$ it holds that $\{X \in \X \mid X \cap S_i \neq \emptyset \}=\{X' \in \X' \mid X' \cap S_i  \neq \emptyset \}$, then $i$ rejects $\X$ if and only if $i$ rejects $\X'$. 
\end{observation}
The above observation captures that an agent $i$ cannot gain anything by rejecting a set that is not intersected by an $i$-internal set. Thus, we can assume they will not reject such sets. $\X$ and $\X'$ become identical after removing these sets, so the observation is correct. \par
Now, we present two reduction rules that act on an input instance~$(U_1,\dots,U_p,\mathcal{S},k)$ of RP-$d$-SP and prove that they are safe. A reduction rule is \emph{safe} if its application preserves the YES/NO answer to the decision problem. The first rule aims to reduce the number of internal sets.

\begin{rrule}\label{rr: internal sets}
    Let $\mathcal{F} \subseteq \mathcal{S}[U_i]$, for some $i \in [p]$, be a sunflower of size $d(k\cdot d-1)+2$.
    Remove a set $F \in \mathcal{F}$ of minimum size.
\end{rrule}

\begin{lemma}\label{lm: rr internal safe}
    \autoref{rr: internal sets} is safe if $\mathcal{S}$ has a hitting set of size  $\le k\cdot d$.
\end{lemma}
\begin{proof}
    Let $I=(U_1,\dots,U_p,\mathcal{S},k)$ be an instance of RP-$d$-SP where there exists some $i \in [p]$ and sunflower $\mathcal{F} \subseteq \mathcal{S}[U_i]$ of size $d(k\cdot d-1)+2$ with core $Y$. Let $S$ be the set of minimum size in $\mathcal{F}$ that is removed by \autoref{rr: internal sets} to create the instance $I'=(U_1,\dots,U_p,\mathcal{S} \setminus \{S\},k)$. 
    We will now argue that $I$ is a YES-instance if and only if $I'$ is a YES-instance.\par
    ($\Rightarrow$)
    Assume $I$ has a solution $\mathcal{X}$. We claim there is a solution for $I$ that does not use $S$. If $S \not \in \mathcal{X}$, our claim holds. Suppose $S \in \mathcal{X}$.
    Invoking \autoref{lm: constructing a different solution} with $h = k\cdot d$ yields that a $S' \in \mathcal{F}$ exists such that the sets in $\X' = (\mathcal{X} \setminus \{S\}) \cup \{S'\}$ are pairwise disjoint.
    Since $S$ is a minimum-size set in $\mathcal{F}$, it follows that $|S'| \ge |S|$, so the collection $\mathcal{X}'$ covers at least~$k$ elements. 
    To show that $\X'$ is rejection-proof, suppose for contradiction that some agent $j$ rejects $\mathcal{X}'$ by proposing the alternate packing $(\X' \setminus \Xrej') \cup \Xint'$. We distinguish three cases: 
    \begin{itemize}
        \item $j\neq i$. 
        Since~$S \in \mathcal{S}[U_i]$, each set $J \in \mathcal{S}$ with $J \cap U_j \neq \emptyset$ is contained in $\mathcal{X}'$ if and only if it is contained in $\mathcal{X}$. From \autoref{ob: idential collections rejection iff} it follows that agent $j$ would reject $\mathcal{X}$.  
        \item $j=i$ and $S' \in \X_{\operatorname{rej}}'$. Let $\X_{\operatorname{rej}}=(\X'_{\operatorname{rej}} \setminus \{S'\}) \cup \{S\}$ and note that $\X' \setminus \X'_{\operatorname{rej}} = \X\setminus \X_{\operatorname{rej}}$. By the definition of rejection  $(\mathcal{X}\setminus \mathcal{X}_{\operatorname{rej}}) \cup \X'_{\operatorname{int}}$ is a collection of disjoint sets.
        Both $S$ and $S'$ are $j$-internal and $|S| \le |S'|$, hence $\X$ covers at most as many elements of $U_j$ as $\X'$ does.  
        By definition of rejection, $ (\X' \setminus \X'_{\operatorname{rej}} ) \cup \X_{\operatorname{int}}'= (\mathcal{X}\setminus \mathcal{X}_{\operatorname{rej}}) \cup \X'_{\operatorname{int}}$ covers more elements from $U_j$ than $\mathcal{X}'$ does, thus also more than $\mathcal{X}$ covers. This would mean that $\mathcal{X}$ would be rejected by agent $j$ by rejecting the sets $\X_{\operatorname{rej}} \subseteq \X$ and replacing these with the sets $\X_{\operatorname{int}}'$.
        \item $j=i$ and $S' \not \in \X_{\operatorname{rej}}'$. Let $\mathcal{X}_{\operatorname{rej}}= \mathcal{X}_{\operatorname{rej}}' \cup \{S\}$ and $\mathcal{X}_{\operatorname{int}} = \mathcal{X}_{\operatorname{int}}' \cup \{S'\}$. Observe that $\mathcal{X}_{\operatorname{int}} \subseteq \mathcal{S}[U_j]$ and $(\X \setminus \X_{\operatorname{rej}}) \cup \X_{\operatorname{int}} = (\X' \setminus \X'_{\operatorname{rej}}) \cup \X'_{\operatorname{int}}$. Hence agent $j$ would reject $\X$.
    \end{itemize}
    As all cases contradict the assumption that $\X$ is rejection-proof, $\mathcal{X}'$ is a solution for $I$ that does not contain $S$. Now \autoref{lm: removing not used sets preserves solution} implies that  $I'$ has a solution and is a YES-instance.\par
    ($\Leftarrow$)
    Assume $I'$ has a solution $\mathcal{X}'$. We claim that $\X'$ is also a solution for the instance $I$. It suffices to show that no rejection can take place where $S \in \X_{\operatorname{int}}$, as the other requirements for $\mathcal{X}'$ to be a solution for $I$ follow directly from $\mathcal{X}'$ being a solution for $I'$.\par
    Assume, for sake of contradiction, that agent $i$ rejects $\X'$ by rejecting $\X_{\operatorname{rej}}$ and replacing these sets with $\X_{\operatorname{int}}$, where $S \in \X_{\operatorname{int}}$. Because $(\X' \setminus \X_{\operatorname{rej}}) \cap \X_{\operatorname{int}}$ is a collection of disjoint sets and $S \in \mathcal{F}$ it follows from \autoref{lm: constructing a different solution} with $h = k\cdot d$ that there exists some set $S' \in \mathcal{F}$ such that $(((\X' \setminus \X_{\operatorname{rej}}) \cap \X_{\operatorname{int}}) \setminus \{S\}) \cup  \{S'\} $ is a collection of disjoint sets.     
    Note that $|S'| \ge |S|$ and $S,S' \in \mathcal{S}[U_i]$. Hence agent $i$ can also reject $\mathcal{X}'$ as a solution for $I$ by rejecting $\X_{\operatorname{rej}}$ and replacing these sets with $(\X_{\operatorname{int}} \setminus \{S\}) \cup \{S'\}$. Similarly, agent $i$ can reject $\X'$ as a solution for $I'$; a contradiction. Hence $\X'$ is a solution for $I$, which means $I$ is a YES-instance. 
\end{proof}

Observe that in the forward direction of the proof we increase the number of sets that are rejected by some agent $j$. This is possible because an agent can reject an unbounded number of sets in RPKE-$d$. The proof as written does not work in the setting of $c$-RPKE-$d$.
Our second reduction rule aims to reduce the number of sets that are not internal.

\begin{rrule}\label{rr: external sets}
    Let $\mathcal{F} \subseteq \mathcal{S}$ be a sunflower with core $Y$ of size $d(k\cdot d-1)+2$ such that:
    \begin{itemize}
        \item No petal $F \setminus Y$, for $F \in \mathcal{F}$, is intersected by an internal set.
        \item For each agent $i$ that has an internal set that intersects $Y$, for all $F,F' \in \mathcal{F}$:
        \begin{equation*}|F \cap U_i| = |F' \cap U_i|.\end{equation*}
    \end{itemize}
    Remove a set $F \in \mathcal{F}$ of minimum size.
\end{rrule}

\begin{restatable}{lemma}{rrexternalsets}\label{lm: rr external safe}
    $(\bigstar)$ \autoref{rr: external sets} is safe if $\mathcal{S}$ has a hitting set of size  $\le k\cdot d$.
\end{restatable}
The proof of the lemma follows the same structure as the proof of \autoref{lm: rr internal safe}. The restrictions placed on the sunflower used in \autoref{rr: external sets} together with \autoref{lm: removing not used sets preserves solution} can be used to prove one direction of the proof. The other direction follows from the fact that any removed set is not internal and can thus not be used as a replacement set in a rejection.


To bound the kernel size we define $f_d(k)=d \cdot d!(d(k\cdot d-1)+1)^d$ and $g_d(k) = {d \cdot d! \left( d^{d-1} \cdot (d(k\cdot d-1)+1) +  d\cdot  k \cdot f_d(k) \right)^d}$. Note that $g_d(k) = \Oh \left( \left(k^{d+1} \right)^d \right)$ for fixed~$d$.

\begin{restatable}{lemma}{rrapplicable} \label{lem:reduction-rule-applicable}
     $(\bigstar)$ If for an instance $I=(U_1,\dots,U_p,\mathcal{S},k)$ of RP-$d$-SP it holds that $|\mathcal{S}| \ge g_d(k)$, then one of the reduction rules can be applied to $I$ in polynomial time, or it can be concluded in polynomial time that $I$ is a YES-instance.
\end{restatable}
\begin{proof}[Proof sketch]
    It can easily be shown that an instance in which at least $k$ agents have an internal set is a YES-instance. Else, if the instance contains at least $k\cdot f_d(k)$ internal sets, we can use \autoref{lm:sunflower} to show that \autoref{rr: internal sets} can be applied in polynomial time. Otherwise, the number of internal sets is bounded, and we show that \autoref{rr: external sets} can be applied. In that case, the number of petals in a sunflower that are intersected by internal sets is also bounded and the core of a sunflower is intersected by internal sets from at most $d$ agents. Since $|\S| \geq g_d(k)$, \autoref{lm:sunflower} yields a sufficiently large sunflower such that some of its sets have the same number of elements from each agent with an internal set intersecting the core. Some subset of this is a sunflower meeting the prerequisites from \autoref{rr: external sets} and this can be found in polynomial time.
\end{proof}

Now, the above insights can be used to give the kernel. Given an RP-$d$-SP instance $I$, it starts by greedily computing an inclusion-wise maximal packing $\X$ of disjoint sets. In the appendix, we show that $I$ is a YES-instance if $|\X| > k$. Otherwise, the union of all sets in~$\X$ forms a hitting set of size $\leq k \cdot d$, meaning that Reduction rules~\ref{rr: internal sets} and~\ref{rr: external sets} are safe to apply. \autoref{lem:reduction-rule-applicable} shows that at least one of them is applicable if $|\mathcal{X}| \geq g_d(k)$ and we repeatedly apply one until an instance is obtained with at most $g_d(k)$ sets. This proves the following.

\begin{restatable}{theorem}{kernel}\label{thm:kernel}
    $(\bigstar)$ There is a polynomial-time algorithm that takes an RP-$d$-SP instance as input and either correctly determines that this is a YES-instance or returns an equivalent RP-$d$-SP instance with at most $g_d(k)$ sets and a hitting set of size at most $k\cdot d$. Hence, for constant $d$, RP-$d$-SP admits a kernel of size $g_d(k) = \Oh \left( \left(k^{d+1} \right)^d \right)$.
\end{restatable}

This kernel can be used as a subroutine to solve RP-$d$-SP. Since it outputs an equivalent instance with a bounded hitting set, the number of collections of pairwise disjoint sets is also bounded. This limits the number of candidate solutions and alternative packings that a brute-force algorithm needs to consider, and it allows us to prove the following statement.

\begin{restatable}{theorem}{bruteforce}\label{thm:fpt-algorithm}
    $(\bigstar)$ RP-$d$-SP is solvable in $2^{\Oh(k \log k)} + (n+m)^{\Oh(1)}$ time for each constant~$d$.
\end{restatable}
\section{A tight ETH lower bound for rejection-proof kidney exchange}
\label{sec:eth-lower-bound}
We prove that the algorithm presented in \autoref{thm:fpt-algorithm} is asymptotically optimal to solve RP-$d$-SP under the ETH for any~$d \geq 3$. More strongly, we show that no faster algorithm exists for the problem RPKE-$3$ under the ETH, even when the input contains only two agents. To see that this is a stronger result, note that when~$d \geq 3$ the RP-$d$-SP problem strictly generalizes RPKE-$3$. The construction also yields a lower bound for $c$-RPKE-$3$ for every~$c \geq 2$. We use the ETH lower bound on \textsc{Subgraph Isomorphism} from \autoref{lem:subgraph-isomorphism} as starting point.

\begin{lemma}
    Assuming the ETH, there is no integer $c \geq 2$ such that $c$-RPKE-$3$ is solvable in $2^{o(n \log n)}$ time, even when the input contains only two agents. 
    The same lower bound applies to RPKE-$3$ with two agents.
\end{lemma}
\begin{proof}
    Let $c \geq 2$ be arbitrary. We prove the statement by a reduction from \textsc{Subgraph Isomorphism}. The resulting instance can serve as an input to both RPKE-$3$ and $c$-RPKE-$3$. First, we give the reduction after which we prove it is correct for both problems. A visual example of the reduction can be seen in \autoref{fig:eth-reduction}.

    \begin{figure}[t]
    \centering
        \begin{minipage}[t]{.561\linewidth}
                \centering
                \includegraphics[width=.772\linewidth]{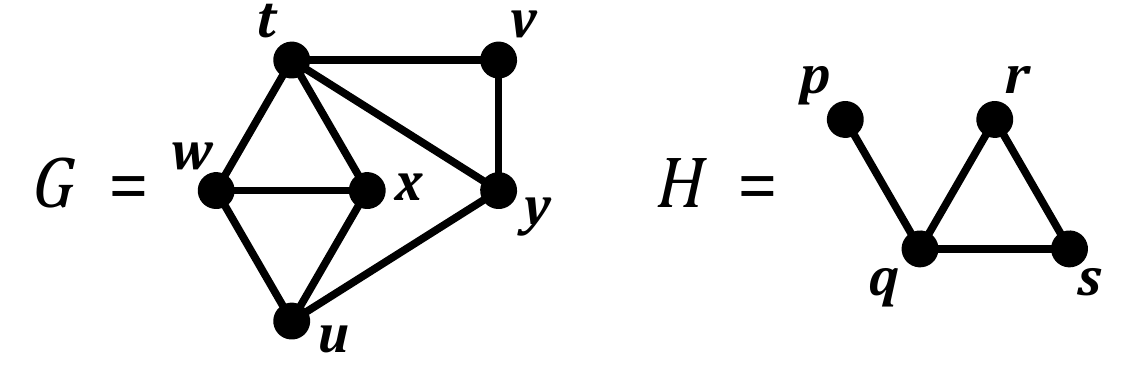}
                \subcaption{An instance $(G,H)$ of \textsc{Subgraph Isomorphism} as example input for the reduction.}
                \label{fig:eth-reduction-input}
        \end{minipage}
        \hfill
        \begin{minipage}[t]{.389\linewidth}
                \centering
                \includegraphics[width=.774\linewidth]{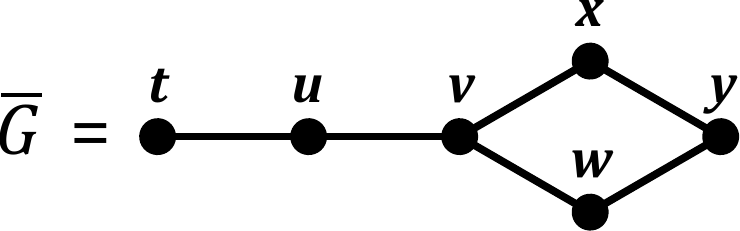}
                \subcaption{The complement of $G$, which is implicitly used to construct $G'$.}
                \label{fig:eth-reduction-complement}
        \end{minipage}
        \\[2\baselineskip]
        \begin{minipage}[t]{.534\linewidth}
                \includegraphics[width=\linewidth]{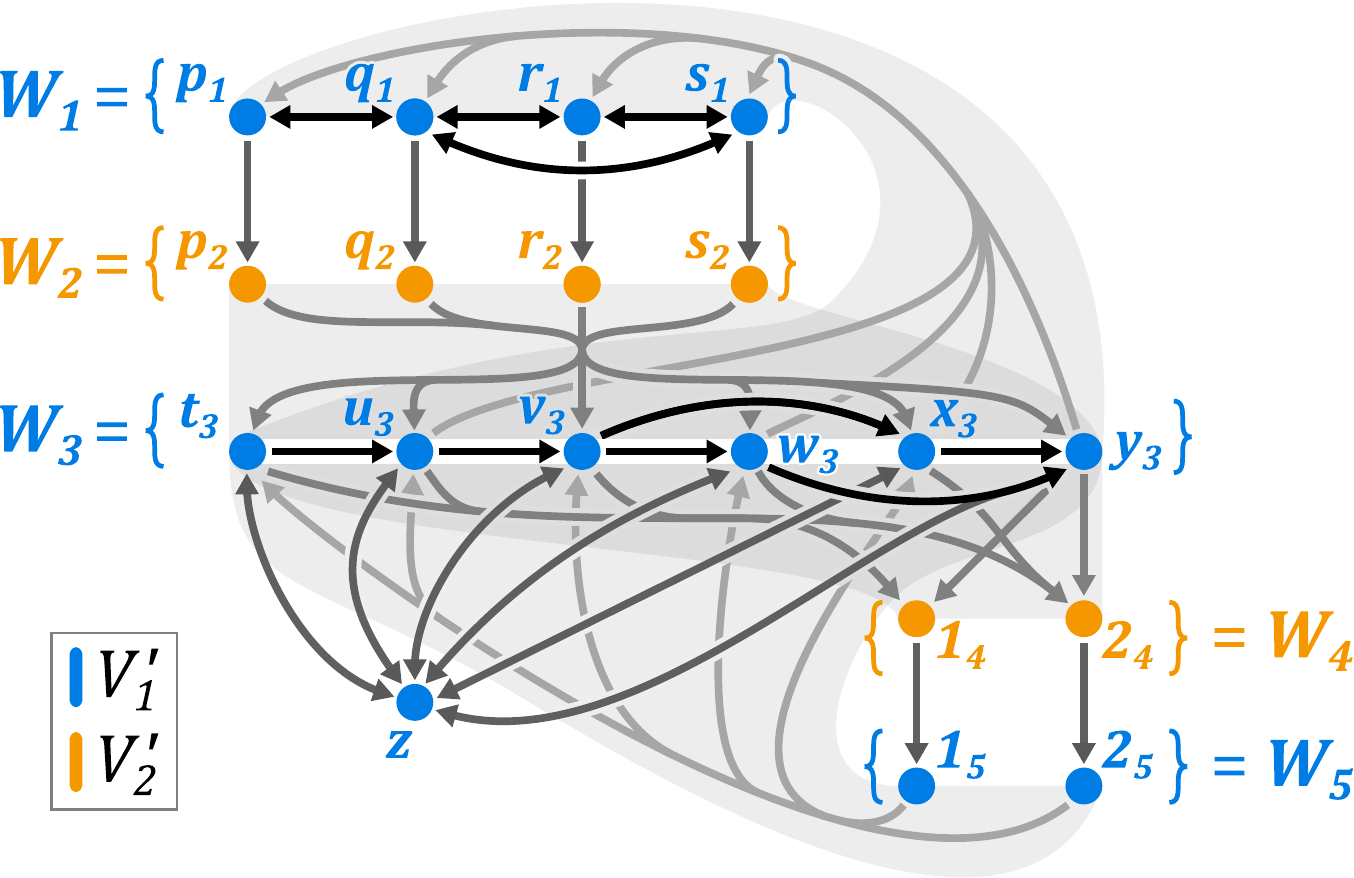}
                \subcaption{Resulting graph $G'$ obtained by running the reduction on the example input shown in \autoref{fig:eth-reduction-input}.}
                \label{fig:eth-reduction-output}
        \end{minipage}
        \hfill
        \begin{minipage}[t]{.407\linewidth}
            \centering
                \includegraphics[width=\linewidth]{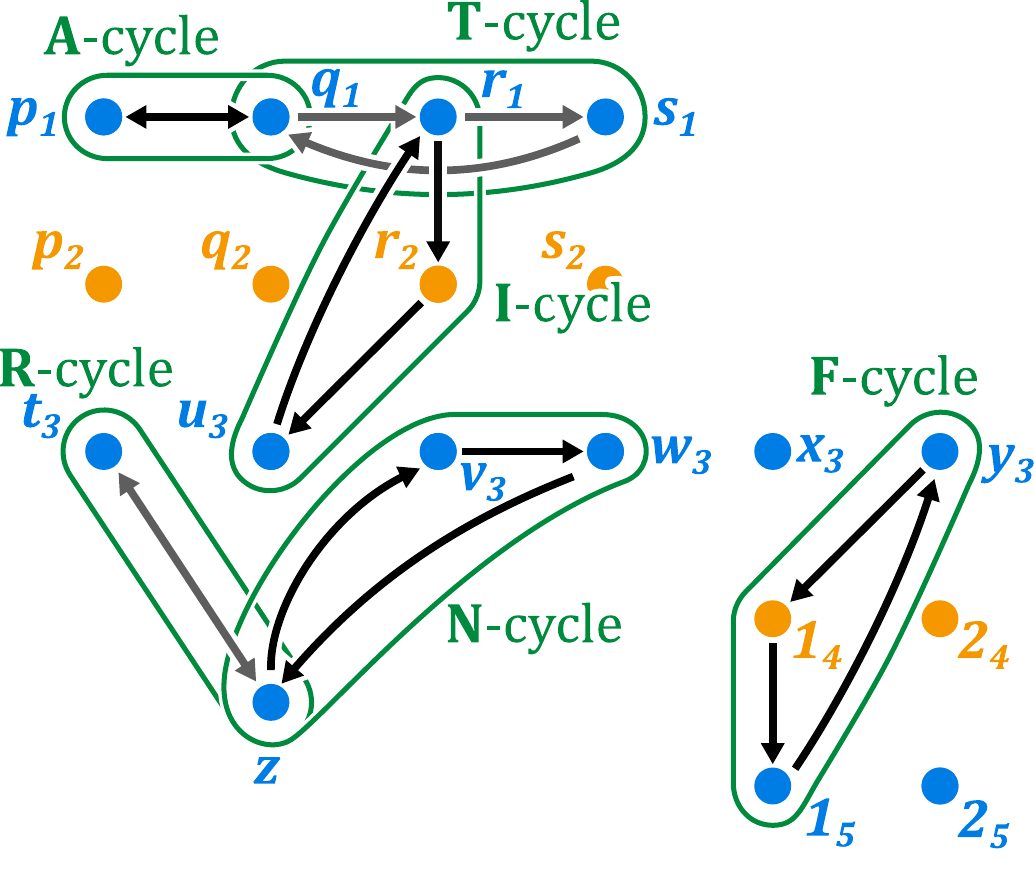}
                \subcaption{A simplified representation of $G'$ indicating each of the different classes of cycles.}
                \label{fig:eth-reduction-cycles}
        \end{minipage}
        \caption{The reduction applied to an example input. For the sake of visualization, not all edges in the graph $G'$ are shown. Four (partially overlapping) areas shaded in gray indicate that all vertices of $W_2$ and $W_5$ each have outgoing edges to all vertices in $W_3$ and that all vertices of $W_3$ have outgoing edges to all vertices of $W_1$ and $W_4$. Vertex colors are used to indicate the partition of the vertex set over the two agents with blue vertices belonging to $V_1'$ and orange vertices to $V_2'$.}
        \label{fig:eth-reduction}
    \end{figure}

    Let~$(G,H)$ be an input for the \textsc{Subgraph Isomorphism} problem consisting of two graphs and let~$n_G := |V(G)|$ and~$n_H := |V(H)|$. We assume that~$n_H \leq n_G$, as~$(G, H)$ would trivially be a NO-instance otherwise. Now, we give our reduction by showing how to modify the input~$(G,H)$ to an input $(G', V_1', V_2', k)$ with two agents for RPKE-$3$ and $c$-RPKE-$3$.

    We start by defining the vertex set of $G'$. To define the complete vertex set, it is useful to think of it as six separate parts: five sets $W_1, W_2, W_3, W_4$, and $W_5$, and a separate vertex~$z$. 
    First, for every vertex $u \in V(H)$, we add a vertex $u_1$ to $W_1$ and a vertex $u_2$ to $W_2$. 
    Next, for every vertex $x \in V(G)$, we add a vertex $x_3$ to $W_3$. Then, for every integer $i \in [n_G - n_H]$ we add a vertex $i_4$ to $W_4$ and a vertex $i_5$ to $W_5$. Finally, we add the additional vertex $z$ to the graph. The vertex set is partitioned over the two agents by defining $V_1' := W_1 \cup W_3 \cup W_5 \cup \{z\}$ and $V_2':= W_2 \cup W_4$. Remark that $|W_1 \cup  W_4| = |W_2 \cup W_5| = |W_3| = n_G$, so that $|V(G')| = 3n_G + 1$.

    Next, we define the edge set of $G'$. Although $G$ and $H$ are undirected graphs, recall that~$G'$, being part of a kidney exchange instance, is a directed graph. 
    \begin{enumerate}
        \item \label{itm:A-edge} For every $\{u,v\} \in E(H)$, we add $(u_1, v_1)$ and $(v_1, u_1)$ to $E(G')$.
        \item \label{itm:I-edge-straight} For every $u \in V(H)$, we add $(u_1, u_2)$ to $E(G')$.
        \item \label{itm:I-edges-other} For every $u \in V(H)$ and every $x \in V(G)$, we add $(u_2, x_3)$ and $(x_3, u_1)$ to $E(G')$.
        \item \label{itm:F-edge-straight} For every $i \in [n_G - n_H]$, we add $(i_4, i_5)$ to $E(G')$.
        \item \label{itm:F-edges-other} For every $x \in V(G)$ and every $i \in [n_G - n_H]$, we add $(x_3, i_4)$ and $(i_5, x_3)$ to $E(G')$.
        \item \label{itm:R-edges} For every $x \in V(G)$, we add $(x_3,z)$ and $(z, x_3)$ to $E(G')$.
        \item \label{itm:N-edge} We fix an arbitrary ordering $\prec$ on $V(G)$. Then, for every $x,y \in V(G)$ such that $\{x,y\} \notin E(G)$ and $x \prec y$, we add $(x_3, y_3)$ to $E(G')$.
    \end{enumerate}
    Finally, we define $k = 3 \cdot n_G$. Note that this is one less than the total number of vertices in~$G'$, denoted as $n := |V(G')| = 3 n_G + 1$. This concludes the description of the reduction. 
    
    We distinguish six different classes of $3$-cycles that live in $G'$. Below, we describe for each of these classes what their intuitive role is in the reduction and we indicate each class of cycles with a letter to refer to them more easily later in the proof. See also \autoref{fig:eth-reduction-cycles}.
    \begin{itemize}
        \item For every edge $\{u,v\} \in E(H)$, the graph $G'$ contains the cycle $\langle u_1, v_1 \rangle$. We call such cycles A-cycles for encoding the \textbf{a}djacencies in $H$.
        \item For every triangle in $H$ on vertices $u$, $v$, and $w$, the graph $G'$ contains both orientations of the directed cycle on vertices $u_1$, $v_1$ and $w_1$. We call such cycles T-cycles due to their correspondence to \textbf{t}riangles in the graph $H$. These cycles originate merely as a byproduct of constructing the A-cycles and serve no explicit purpose in our reduction.
        \item For every $u \in V(H)$ and every $x \in V(G)$, the graph $G'$ contains the cycle $\langle u_1, u_2, x_3 \rangle$. We call such cycles I-cycles, as they are used to encode subgraph \textbf{i}somorphisms from $H$ to $G$: including a cycle $\langle u_1, u_2, x_3 \rangle$ will correspond to mapping $u \in V(H)$ to $x \in V(G)$ in a subgraph isomorphism. See also \autoref{fig:eth-reduction-isomorphisms} (left image) for an example.
        \item For every $x \in V(G)$ and every $i \in [n_G - n_H]$, the graph $G'$ contains the cycle $\langle x_3, i_4, i_5 \rangle$. We call such cycles F-cycles as they are used in the construction of optimal solutions as \textbf{f}iller cycles to cover vertices in $W_3$ that are not already covered by I-cycles.
        \item For every $x,y \in V(G)$ with $\{x,y\} \notin E(G)$, the graph $G'$ contains the cycle $\langle x_3, y_3, z \rangle$ (or $\langle y_3, x_3, z \rangle$ depending on whether $x \prec y$). We call such cycles N-cycles for encoding the \textbf{n}on-adjacencies in $G$.
        \item Finally, for every $x \in V(G)$, the graph $G'$ contains the cycle $\langle x_3, z \rangle$. Much like T-cycles, these cycles serve no explicit purpose in our reduction and they originate merely as byproduct of constructing the N-cycles. We call them R-cycles, since optimal rejection-proof solutions typically do not include any of them, making these cycles mostly \textbf{r}edundant.
    \end{itemize}

    \begin{figure}[t]
    \centering
        \begin{minipage}[t]{.475\linewidth}
                \centering
                \includegraphics[width=.779\linewidth]{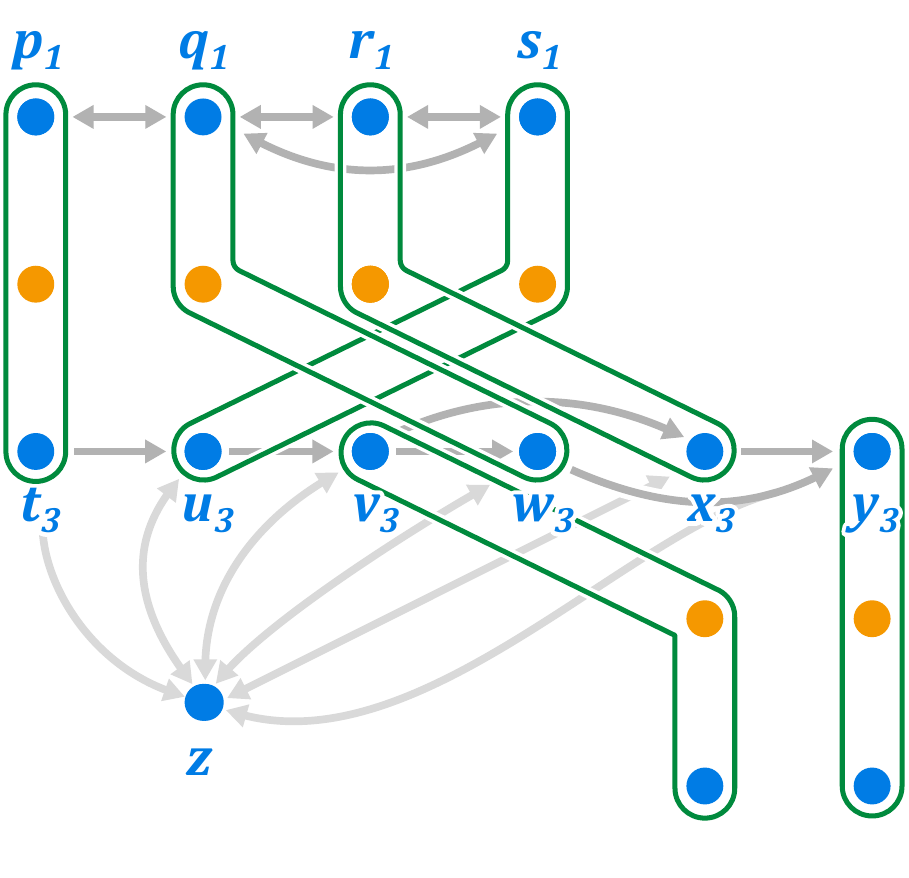}
        \end{minipage}
        \hfill
        \begin{minipage}[t]{.475\linewidth}
                \centering
                \includegraphics[width=.779\linewidth]{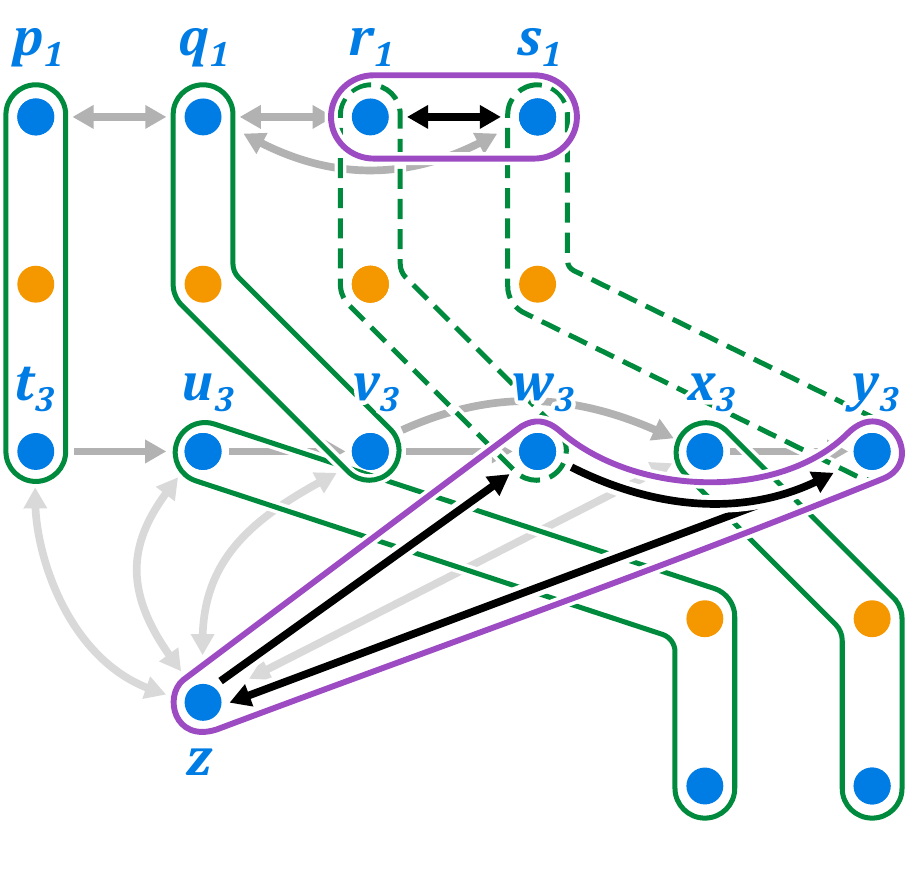}
        \end{minipage}
        \caption{Both figures depict the same graph $G'$ as in \autoref{fig:eth-reduction}, each with a different packing of cycles shown in green. For readability, only some edges are shown. The cycle packing on the left encodes the subgraph isomorphism from $H$ to $G$ that maps $p$ to $t$, $q$ to $w$, $r$ to $x$, and $s$ to $u$. 
        The green packing on the right encodes a function that maps $r$ to $w$ and $s$ to $y$. This is not a subgraph isomorphism because $\{r, s\} \in E(H)$ and $\{w , y\} \notin E(G)$. Hence, agent $1$ can reject the green packing by proposing an alternative one in which the dashed cycles are replaced by the purple cycles.}
        \label{fig:eth-reduction-isomorphisms}
    \end{figure}

    Before proving the correctness of our reduction, observe that $G'$ does not contain any $3$-cycles other than those contained in the six classes described above. In particular, note that $G'[W_3]$ is a DAG as the obvious extension of $\prec$ to $W_3$ forms a topological ordering. Finally, cycles whose length is greater than $3$ will not be relevant in our proof, since RPKE-$3$ and $c$-RPKE-$3$ only deal with $3$-cycles by definition.

    Most of the remainder of this proof is used to show the correctness of the reduction. We do this by proving that $H$ is isomorphic to a subgraph of $G$ if and only if $G'$ has a ($c$-)rejection-proof packing of $3$-cycles that covers at least $k$ vertices. Below, we prove the two directions of this bi-implication separately.

    \begin{claim} \label{clm:complete}
        If $G'$ contains a $2$-rejection-proof packing of $3$-cycles that covers at least $k = 3 \cdot n_G$ vertices (e.g. because this packing is even (completely) rejection-proof), then $H$ is isomorphic to a subgraph of $G$.
    \end{claim}
    \begin{claimproof}
        Let $\C$ be a $2$-rejection-proof packing of $3$-cycles in $G'$ that covers at least $k = 3 \cdot n_G$ vertices. Before showing how to construct from $\C$ a subgraph isomorphism from $H$ to $G$, we derive some structural properties of $\C$. We start by observing that at most one vertex from~$G'$ is not covered by~$\C$. We use this property to show that $\C$ must contain $n_H$ I-cycles. 
        
        First, we argue that the number of I-cycles in $\C$ must be more than $n_H - 2$. Since vertices from $W_2$ only appear in I-cycles and each I-cycle contains exactly one of these vertices, at least two vertices from $W_2$ would be uncovered by $\C$ otherwise. This would contradict the assumption that $\C$ leaves at most one vertex uncovered.

        Next, we argue that the number of I-cycles in $\C$ cannot be exactly $n_H - 1$. In that case exactly one vertex $u_1$ from $W_1$ and exactly one vertex $v_2$ from $W_2$ would not be covered by an I-cycle in $\C$. If $v_2$ is not covered by an I-cycle in $\C$, this means that it is not covered at all by $\C$. Furthermore, $u_1$ only appears in I-cycles, A-cycles, and T-cycles. By assumption, it is not covered by an I-cycle in $\C$. Moreover, it cannot be covered by an A-cycle or T-cycle in $\C$ either: such cycles contain at least two vertices from $W_1$, but $u_1$ is the only vertex in $W_1$ that is not covered by an I-cycle in $\C$. As such, $\C$ does not contain A-cycles or T-cycles since these would not be vertex-disjoint from all I-cycles in $\C$. This means that neither $v_2$ nor $u_1$ is covered by $\C$, contradicting that $\C$ leaves at most one vertex uncovered.

        By arguing that the number of I-cycles in $\C$ must be strictly more than $n_H - 2$ but not exactly $n_H - 1$, we have shown that the number of I-cycles in $\C$ is exactly $n_H$. This means that all vertices from $W_1$ and $W_2$ are covered by I-cycles in $\C$.

        Next, we see that $\C$ must include $n_G - n_H$ F-cycles. If it were to contain less F-cycles, at least one vertex from $W_4$ and $W_5$ each would not be covered by F-cycles. Since vertices from these two sets appear only in F-cycles, this would mean that they are not covered by $\C$ at all. In turn, this would contradict the observation that leaves $\C$ at most one vertex uncovered.

        By including $n_H$ I-cycles and $(n_G-n_H)$ F-cycles, all vertices of $G'$ except $z$ are covered. Since every cycle contains a vertex different from $z$, the packing $\C$ does not contain any other cycles. Knowing the structure of $\C$ allows us to construct a subgraph isomorphism $\varphi: V(H) \rightarrow V(G)$.

        Let $u \in V(H)$ be arbitrary, and let $C_u \in \C$ be the unique I-cycle from $\C$ that contains $u_1$. Since every vertex from $W_1$ is contained in exactly one I-cycle in $\C$, this is well-defined. Let $x_3 \in W_3$ be the unique vertex from $W_3$ in $C_u$. Then, we define $\varphi(u) = x$. To complete the construction of $\varphi$, we do the same for every $u \in V(H)$.

        This construction yields a function~$\varphi$ that maps every vertex in~$H$ to a unique vertex in~$G$. To show that $\varphi$ is in particular a subgraph isomorphism, it remains to show that for every $\{u,v\} \in E(H)$ it holds that $\{\varphi(u), \varphi(v)\} \in E(G)$. Suppose for contradiction that there is some $\{u,v\} \in E(H)$ for which this is not the case, i.e.~$\{\varphi(u), \varphi(v)\} \notin E(G)$. We show that agent $1$ would then reject $\C$. The right image in \autoref{fig:eth-reduction-isomorphisms} shows an example of this situation.

        Let $x = \varphi(u)$ and $y = \varphi(v)$. Assume w.l.o.g. that $x \prec y$ in the arbitrary ordering fixed onto $V(G)$ in step \ref{itm:N-edge} of the construction of $E(G')$. Then the following cycles exist:
        \begin{itemize}
            \item Because $\{u,v\} \in E(H)$, the A-cycle $\langle u_1, v_1 \rangle$ is added to~$G'$ in construction step \ref{itm:A-edge}.
            \item Because $\{x,y\} \notin E(G)$, the N-cycle $\langle x_3, y_3, z \rangle$ is added to~$G'$ in steps~\ref{itm:R-edges} and~\ref{itm:N-edge}.
        \end{itemize}
        Now, agent $1$ can reject the solution by proposing an alternative packing to $\C$ in which $\langle u_1, u_2, x_3 \rangle$ and $\langle v_1, v_2, y_3 \rangle$ are removed and the cycles $\langle u_1, v_1 \rangle$ and $\langle x_3, y_3, z \rangle$ are added. 
        
        This modification to $\C$ yields a packing of cycles that are still pairwise disjoint and which cover all vertices of $V_1'$. Moreover, the two cycles added by agent $1$ are indeed $1$-internal cycles in the graph, which means that agent $1$ will indeed reject $\C$ whenever there is a $\{u,v\} \in E(H)$ such that $\{\varphi(u), \varphi(v)\} \notin E(G)$. This contradicts the assumption that $\C$ is 2-rejection-proof, so we conclude that $\varphi$ is a valid subgraph isomorphism from $H$ to $G$.
    \end{claimproof}
    \begin{restatable}{claim}{ReductionSound} \label{clm:reduction-sound}
        ($\bigstar$) If $H$ is isomorphic to a subgraph of $G$, then $G'$ contains a rejection-proof packing of $3$-cycles (which is thus also $c$-rejection-proof) that covers $k = 3 \cdot n_G$ vertices.
    \end{restatable}
    \begin{claimproof}[Proof sketch]
        Let $\varphi:V(H) \rightarrow V(G)$ be a subgraph isomorphism from $H$ to $G$. To encode this isomorphism as a packing of $3$-cycles, we include the I-cycle $\langle u_1, u_2, x_3 \rangle$ for every $u \in V(H)$ and $x \in V(G)$ such that $x = \varphi(u)$. We extend this to a packing $\C$ that covers~$k$ vertices by adding arbitrary but disjoint F-cycles that cover the $n_G - n_H$ vertices from $W_3$ that are not contained in any of the added I-cycles. It remains to show that this packing is rejection-proof.
        
        As $\C$ leaves only the vertex $z$ uncovered, which belongs to $V_1'$, only agent $1$ can reject and must do so by constructing an alternative packing that covers all vertices of $V_1'$. In particular, they must propose to add some internal cycle $C_z$ that includes $z$. By construction of $G'$, this cycle is either an R-cycle or an N-cycle and we can show with a case distinction that neither of those two options extends to a valid rejection for agent $1$.
    \end{claimproof}
    
    Finally, we argue that Claims~\ref{clm:complete} and~\ref{clm:reduction-sound} combine to prove the original lemma statement. Together, they prove the correctness of the reduction from \textsc{Subgraph Isomorphism} to both RPKE-$3$ and $c$-RPKE-$3$. Moreover, the reduction can be executed in polynomial time. Hence, if an algorithm were to exist that solves either RPKE-$3$ or $c$-RPKE-$3$ in $2^{o(n \log n)}$ time, this algorithm could be used to determine in $2^{o(n \log n)} = 2^{o(n_G \log n_G)}$ time whether $H$ is isomorphic to a subgraph of $G$. By \autoref{lem:subgraph-isomorphism}, this is not possible assuming the ETH.
\end{proof}

For completeness, we remark without proof that the construction above can be modified to prove that for any $d \geq 3$ and~$c\geq 3$, an algorithm that solves $c$-RPKE-$d$ in time~$2^{o(n \log n)}$ would violate the ETH.

\section{A single-exponential algorithm for 1-rejections}
\label{sec:1-rejection-algo}
In the previous section, we have shown that under the ETH it is impossible to solve RPKE-$d$ (or its restricted case $2$-RPKE-$d$) in time~$2^{o(n \log n)}$. This raises the question whether a single-exponential running time is possible for $1$-RPKE-$d$. We answer this question positively by providing a $3^n\cdot(n + m)^{\Oh(1)}$ algorithm for the generalization $1$-RP-$d$-SP.

\begin{restatable}{theorem}{thmOneRejection}\label{thm:1-rejection-algo}
    $(\bigstar)$ For each fixed~$d$, the $1$-RP-$d$-SP problem can be solved in $3^n\cdot(n + m)^{\Oh(1)}$ time.
\end{restatable}
\begin{proof}[Proof sketch]
    The main idea for solving an instance $I = (U_1,\dots,U_p,\mathcal{S},k)$ of $1$-RP-$d$-SP is as follows. Let~$U := \bigcup_{i \in [p]} U_i$ be the complete universe. The algorithm iterates over all subsets~$U' \subseteq U$ of size at least~$k$, each of which could be the set of elements covered by a solution to~$I$. It then tries to find a solution covering exactly~$U'$. The main insight now is the following: whether or not a set~$A \in \S[U']$ leads to a $1$-rejection by agent~$i$ in a solution that covers exactly~$U'$, is uniquely determined by the combination of~$U'$ and~$A$. It is independent of the identities of the other sets used in the solution. Agent~$i$ performs a $1$-rejection for~$A$ if and only if there is a collection of pairwise disjoint $i$-internal sets avoiding~$U' \setminus A$ that covers more elements from~$U_i$ than~$A$ does. For each choice of~$U'$ we can therefore compute a collection~$\mathcal{S}'$ of sets that are feasible to be used in a $1$-rejection-proof solution covering exactly~$U'$. Testing whether there is a $1$-rejection-proof solution covering exactly~$U'$ then turns into an instance of \textsc{Exact Set Cover} over a universe of size~$|U'| \leq |U|$, which can be solved efficiently via \autoref{lm:exact set cover}. The exact running time follows from some small optimizations and the binomial theorem.
\end{proof}

\section{Conclusion and discussion} \label{sec:conclusion}
In this paper we started the parameterized complexity analysis of rejection-proof kidney exchange, which was recently introduced by Blom et al.~\cite{Blom24} due to practical applications. Our FPT algorithm shows that the complexity of the problem is governed by the size of the solution, rather than the size of the instance. For the single-agent version of the kidney exchange problem, its parameterized complexity has also been considered with respect to structural parameterizations such as treewidth. We leave such investigations of rejection-proof kidney exchange to further work. Looking beyond the kidney exchange (i.e., $d$-cycle packing) problem studied in this paper, one could consider rejection-proof versions of other packing problems such as \textsc{Chordless Cycle Packing}~\cite{Marx20} and \textsc{Odd Cycle Packing}~\cite{KawarabayashiR10}. 


\bibliography{references}

\clearpage 

\appendix
\section{Omitted proofs from Section 2}

\exactSetCoverLemma*
\begin{proof}[Proof of \autoref{lm:exact set cover}]
    The algorithm follows by a standard application of dynamic programming over subsets~\cite[Section 6.1]{CyganFGKMPS16}. Given a set system~$(U, \mathcal{S})$, we solve a sequence of subproblems to determine whether~$U$ can be covered by a collection of pairwise disjoint sets from~$\mathcal{S}$. Number the sets in~$\mathcal{S}$ as~$\mathcal{S} = \{S_1, \ldots, S_m\}$. For a subset of elements~$U'\subseteq U$ and integer~$0 \leq i \leq m$, let~$T[U', i]$ be the Boolean value that indicates whether there is a collection of pairwise disjoint sets from~$\{S_1, \ldots, S_i\}$ whose union is exactly~$U'$. Hence the overall aim of the algorithm is to determine~$T[U,m]$. It is easy to see that the following recurrence holds for~$T$:
    \begin{equation}
        T[U', i] = \begin{cases}
            \mathsf{true} & \text{if $i=0$ and $U' = \emptyset$} \\
            \mathsf{false} & \text{if $i=0$ and $U' \neq \emptyset$} \\
            T[U', i-1] \vee (S_i \subseteq U' \wedge T[U' \setminus S_i, i-1]) & \text{if $i>0$.}
        \end{cases}
    \end{equation}
    By computing the values for~$T$ by increasing size of~$U'$ and increasing value of~$i$, we compute~$T[U, m]$. The number of subproblems is~$2^{|U|} \cdot (m+1)$. Since each subproblem can be handled in polynomial time, the overall algorithm can be implemented to run in time~$2^{|U|} \cdot (|U|+|\S|)^{\Oh(1)}$.
\end{proof}

\differentsolution*
\begin{proof}
    We prove the statement by showing that one of the sets in $\mathcal{F} \setminus \{X\}$ is disjoint from all sets in $\X \setminus \{X\}$. Since $\X$ is a set of pairwise disjoint sets and $X$ contains the entire core of $\mathcal{F}$, the core is disjoint from all sets $\X \setminus \{X\}$. Therefore, it remains to show that the sunflower contains a petal that is disjoint from all sets in $\X \setminus \{X\}$ (and which is not the petal belonging to $X$ itself).

    To this end, we first argue that $|\X| \leq h$. Every hitting set must contain an element of each set in $\X$. Since the set system has a hitting set of size $\leq h$ and $\X$ is a set of pairwise disjoint sets, it follows that $|\X| \leq h$.

    Every set in $\S$ (in particular also those in $\X$) is of size at most $d$ so that $\X \setminus \{X\}$ covers at most $d(h - 1)$ elements. Hence, $\X \setminus \{X\}$ covers at most $d(h-1)$ petals of $\mathcal{F}$ as these are all pairwise disjoint by definition of a sunflower. This leaves at least two petals of $\mathcal{F}$ disjoint from $\X \setminus \{X\}$. One of these is the petal that belongs to $X$, leaving at least one petal belonging to some other set $X' \subseteq \mathcal{F}\setminus\{X\}$ that is disjoint from $\S \setminus \{X\}$.
\end{proof}
\section{Omitted proofs from Section 3}
\rrexternalsets*
\begin{proof}
    Let $I=(U_1,\dots,U_p,\mathcal{S},k)$ be an instance of RP-$d$-SP where \autoref{rr: external sets} is applicable. Let $\mathcal{F}$ be a sunflower with core $Y$ as defined in this reduction rule and let $S \in \mathcal{F}$ be the minimum-size set that is removed by applying \autoref{rr: external sets} to obtain the instance $I'=(U_1,\dots,U_p,\mathcal{S} \setminus \{S\},k)$. We proceed by arguing that $I$ is a YES-instance if and only if~$I'$ is a YES-instance. \par
    ($\Rightarrow$)
    Assume $I$ to be a YES-instance, then there must exist a solution $\mathcal{X}$ for the instance $I$. We claim there exists a solution for $I$ that does not use $S$. If $S \not \in \mathcal{X}$ our claim holds. Suppose $S \in \mathcal{X}$. By \autoref{lm: constructing a different solution} with $h = k\cdot d$, there exists a set $S' \in \mathcal{F}$ such that $\X' = (\mathcal{X} \setminus \{S\}) \cup \{S'\}$ is a collection of pairwise disjoint sets. Since $\mathcal{X}$ covers at least $k$ elements, while $S$ and~$S'$ both belong to~$\mathcal{F}$ and $S$ is a set of minimum size in $\mathcal{F}$, it follows that $\X'$ covers at least $k$ elements. \par
    It remains to prove that no agent rejects $\X'$. Assume, for sake of contradiction, that some agent $i \in [p]$ rejects $\mathcal{X}'$ by rejecting $\X_{\operatorname{rej}}$ and replacing these sets with $\mathcal{X}_{\operatorname{int}}$. 
    We distinguish two different cases:
    \begin{itemize}
        \item Suppose there is no $i$-internal set intersecting the core $Y$. The petals of $S'$ and $S$ are not intersected by any internal set, so also not any $i$-internal set. Hence any set intersected by an $i$-internal set is contained in $\mathcal{X}$ if and only if it is contained in  $\X'$. Because agent~$i$ rejects  $\X'$ it must also reject $\mathcal{X}$, which follows from \autoref{ob: idential collections rejection iff}.
        \item Suppose there exists an $i$-internal set intersecting $Y$. By the definition of rejection,  
        $(\X'\setminus \X_{\operatorname{rej}}) \cup \X_{\operatorname{int}}$ covers more elements of agent $i$ than $\X'$ does. From the definition of
        \autoref{rr: external sets} it follows that $|S' \cap U_i|=|S \cap U_i|$.  Therefore, $(\X'\setminus \X_{\operatorname{rej}}) \cup \X_{\operatorname{int}}$ also covers more elements of agent $i$ than the original solution $\X$ does.
        \begin{itemize}
            \item If $S' \in \X_{\operatorname{rej}}$, then we consider $\X_{\operatorname{rej}}' = \X_{\operatorname{rej}} \setminus \{S'\} \cup \{S\}$. Observe that $\X \setminus \X_{\operatorname{rej}}' = \X' \setminus \X_{\operatorname{rej}}$ and $\X_{\operatorname{rej}}' \subseteq \X$. Agent~$i$ therefore rejects $\X$ by rejection $\X_{\operatorname{rej}}'$ and replacing them with $\X_{\operatorname{int}}$. 
            \item If $S' \notin \X_{\operatorname{rej}}$, then $\X_{\operatorname{int}} \cup \{S'\}$ is a collection of disjoint sets meaning that no set in $\X_{\operatorname{int}}$ intersects the core~$Y$. The petal $S \setminus Y$ cannot be intersected by any set in $\X_{\operatorname{int}}$, because $S \in \mathcal{F}$ and all the sets in $\X_{\operatorname{int}}$ are internal.
            Hence $(\X \setminus (\X_{\operatorname{rej}} \setminus \{S'\})) \cup \X_{\operatorname{int}}$ is a collection of disjoint sets and covers the same number of elements of $U_i$ as $(\X' \setminus \X_{\operatorname{rej}}) \cup \X_{\operatorname{int}}$. Consequently, agent $i$ rejects $\X$ by rejecting $\X_{\operatorname{rej}} \setminus \{S'\} \subseteq \X$ and replacing these with $\X_{\operatorname{int}}$.
        \end{itemize}
    \end{itemize}
    Since all cases lead to a contradiction, namely that agent $i$ would reject solution $\X$, we infer that $\mathcal{X}'$ is a solution for $I$ which does not contain $S$.
    From \autoref{lm: removing not used sets preserves solution} it then follows that $\X'$ is a solution for $I'$, hence $I'$ is a YES-instance.\par
    ($\Leftarrow$)
    Assume $I'$ to be a YES-instance, then there must exist a solution $\mathcal{X}'$ for the instance $I'$. Because $S$ is not an internal set it cannot be used in a collection of replacement sets $\X_{\operatorname{int}}$. Any rejection of $\X'$ in the instance $I$ would mean that it would also be rejected in the instance $I'$. Hence the collection $\X'$ is a solution for $I$, which implies $I$ is a YES-instance.
\end{proof}

\rrapplicable*
\begin{proof}
    We claim that if there exists some set of agents $A \subseteq [p]$ where $|A| \ge k$ and for each $a \in A$ it holds that $\mathcal{S}[U_a] \neq \emptyset$, then $I$ is a YES-instance. The union of maximum packings of every $(U_i, \S[U_i])$ is a rejection-proof collection of disjoint sets and contains at least $|A| \ge k$ elements. A solution would exists for $I$ meaning it is a YES-instance.
    Observe that checking whether such a set $A$ exists can be done in polynomial time.  \par
    If such a set $A$ does not exist, we check if there are at least $k \cdot f_d(k)$ internal sets. 
    If this is the case there exists an agent $i$ such that $|\mathcal{S}[U_i]| \ge f_d(k)$. From \autoref{lm:sunflower} it follows that we can find a sunflower $\mathcal{F} \subseteq \mathcal{S}[U_i]$ containing $d(k\cdot d-1)+2$ sets in polynomial time. We can perform \autoref{rr: internal sets} on $\mathcal{F}$.
    \par
    If there are less than $k \cdot f_d(k)$ internal sets and less than $k$ agents with an internal set our aim is to use \autoref{rr: external sets}. 
    From \autoref{lm:sunflower} it follows that a sunflower $\mathcal{F} \subseteq \mathcal{S}$ containing $d^{d-1} \cdot (d(k\cdot d-1)+1) +  d\cdot  k \cdot f_d(k) +1$  sets can be found in polynomial time.  At most $d\cdot  k \cdot f_d(k)$ petals can be intersected by an internal set. Create $\mathcal{F}'$ by keeping only those sets in $\mathcal{F}$  whose petal is not intersected by internal sets, it contains at least $d^{d-1} \cdot (d(k\cdot d-1)+1)+1$ sets. \par
    No $i$-internal set intersects an $j$-internal set if $i \neq j$ and the core $Y$ of $\mathcal{F}'$ contains at most $d-1$ elements, hence the set $C$ containing agents that have an internal set intersecting the core contains at most $d-1$ agents. For each of these agents $i \in C$ and each set $S \in \mathcal{F}'$ the value $|S \cap U_i|$ is an integer value between $1$ and $d$. In polynomial time we can create equivalence classes such that for each equivalence class $\sigma$ it holds for each combination of sets $S,S' \in \sigma$ and $i \in C$ that $|S \cap U_i|= |S' \cap U_i|$. There can be at most $d^{d-1}$ of these equivalence classes, so there exists an equivalence class $\mathcal{F}''$ which contains at least $d(k\cdot d-1)+2$ sets. Observe that $\mathcal{F}''$ is a sunflower on which we can perform \autoref{rr: external sets}.        
\end{proof}

\kernel*
\begin{proof}
    The safety of both reduction rules, together with \autoref{lem:reduction-rule-applicable}, almost directly implies the existence of a kernel. However, note that the safety of the reduction rules requires the additional precondition that the input contains a hitting set of size $\leq k\cdot d$. Therefore, most of this proof is used to show that we can identify all instances that do not have such a hitting set as YES-instances in polynomial time.

    Now, to define the kernel, let $(U_1,\dots,U_p,\mathcal{S},k)$ be the RP-$d$-SP instance received as input and suppose it contains at least $g_d(k)$ sets. We start by greedily constructing a maximal packing $\X \subseteq \S$ of pairwise disjoint sets. Next, we consider two cases based on the size of $\X$.

    For the first case, suppose that $|\X| \leq k$. Then, the union $H$ of all sets in $\X$ forms a hitting set to the given set system: by the maximality of $\X$, no set in $\S$ is disjoint from~$H$. Since all sets are of size $\leq d$, we find that the hitting set $H$ is of size $\leq k\cdot d$. Hence, the precondition to the safety of both reduction rules is satisfied. Moreover, since we assumed the input to contain at least $g_d(k)$ sets, \autoref{lem:reduction-rule-applicable} tells us that one of these reduction rules can be applied to the input in polynomial time or that we can conclude in polynomial time that the given instance is a YES-instance. Both reduction rules guarantee to remove at least one set from $\S$, making the input strictly smaller while maintaining the property that the instance has a hitting set of size $\leq k\cdot d$.

    For the second case, suppose that $|\X| > k$. We show that the input is a YES-instance in this case. To this end, suppose that $\X^\ast$ is an inclusion-wise maximal rejection-proof packing of disjoint sets in the given instance. That is, no strict superset of $\X^\ast$ is a rejection-proof packing of disjoint sets. Recall that there is always at least one such packing, as one can be constructed as the union of maximum packings of every $(U_i, \S[U_i])$ to be rejection-proof, greedily extended to be inclusion-wise maximal. Now, suppose for contradiction that $\X^\ast$ covers strictly less than $k$ elements. Then, since $|\X| > k$, there is at least one set $X \in \X$ that is disjoint from all sets in $\X^\ast$. As such, the collection $\X^\ast \cup \{X\}$ is still a packing of pairwise disjoint sets.

    Since we assumed $\X^\ast$ to be rejection-proof, the larger packing $\X^\ast \cup \{X\}$ is also rejection-proof: if there was an agent that would reject $\X^\ast \cup \{X\}$ by proposing the alternative packing $((\X^\ast \cup \{X\}) \setminus \Xrej) \cup \Xint$, then this agent would also reject $\X^\ast$ by proposing the alternative packing $(\X^\ast \setminus (\Xrej \setminus \{X\})) \cup \Xint$ contradicting that $\X^\ast$ is rejection-proof. This indeed shows $\X^\ast \cup \{X\}$ to be rejection-proof, in turn contradicting the assumption that $\X^\ast$ was an inclusion-wise maximal such set. We conclude that $\X^\ast$ covers at least $k$ elements, making the input instance a YES-instance.

    This insight yields the final ingredient for the kernel, which can now be obtained by repeatedly applying the above. That is, greedily computing a maximal packing $\X$ of disjoint sets, and, depending on the size of $\X$, either apply one of the reduction rules or conclude that $\X$ is a YES-instance. Repeating this procedure until the input contains less than $g_d(k)$ sets (in which case it also has a hitting set of size $\leq k\cdot d$) or until it has been correctly identified as a YES-instance thus yields a kernel of size~$g_d(k)$. Treating $d$ as a constant, $g_d(k)$ simplifies to $\Oh \left( \left(k^{d+1} \right)^d \right)$.
\end{proof}

\bruteforce*
\begin{proof}
    Given an RP-$d$-SP instance $I_1$ with $n$ elements and $m$ sets, we start by executing the kernel from \autoref{thm:kernel} in $(n+m)^{\Oh(1)}$ time. Afterwards, we have either correctly concluded that the input is a YES-instance, or we have obtained an equivalent RP-$d$-SP instance $I_2 = (U_1, \ldots, U_p, \S, k)$ with $\Oh\left( (k^{d+1})^d \right)$ sets and a hitting set of size at most $k\cdot d$. Furthermore, we assume that all elements in $I_2$ appear in at least one set of $\S$ or we spend $(n+m)^{\Oh(1)}$ time to remove all elements that do not appear in any set otherwise.

    Next, we describe a brute-force algorithm that determines whether there is a rejection-proof packing of pairwise disjoint sets in $I_2$ that covers at least $k$ elements. Since $I_2$ contains a hitting set of size $\leq k\cdot d$, a packing of pairwise disjoint sets must also be of size $\leq k\cdot d$. The algorithm iterates over all subsets $\X$ of $\S$ that consist of at most $k\cdot d$ sets. Then, for every such collection $\X$, it checks in polynomial time whether it is a collection of pairwise disjoint sets and whether it covers at least $k$ elements. If both those properties hold, it remains to check whether $\X$ is rejection-proof.

    If there is an agent $i$ that rejects the candidate solution $\X$ by proposing to add some set of internal cycles $\Xint$, then just from the contents of $\Xint$ we can deduce a suitable set $\Xrej$ of corresponding rejections and therefore the alternative packing $(\X \setminus \Xrej) \cup \Xint$. This packing can be constructed from $\X$ by removing all sets from $\X$ that intersect $\Xint$ and adding the sets in $\Xint$. In particular, there is no need to consider rejecting more cycles. Hence, to check whether there is an agent $i$ that rejects $\X$, it suffices to check whether there is a collection $\Xint$ of internal cycles that can be used to construct an alternative packing as described above. Since such a set $\Xint$ must be a collection of pairwise disjoint sets, we recall that it is of size at most $k\cdot d$. As such, the algorithm iterates over all subsets $\X' \subseteq \S$ of size at most $k\cdot d$ and checks for every one of them whether the following three properties hold.
    \begin{itemize}
        \item Whether $\X'$ consists of $i$-internal sets for some agent $i$.
        \item If so, whether $\X'$ is a collection of pairwise disjoint sets.
        \item If so, let $\Xrej \subseteq \X$ be the collection of sets in $\X$ that intersect $\X'$ and check whether $(\X \setminus \Xrej) \cup \X'$ covers more elements of agent $i$ than $\X$ does.
    \end{itemize}
    Clearly, all three properties can be checked in polynomial time. If there is a set $\X'$ that satisfies all these properties, then $\X$ is not rejection-proof. If there is no such set $\X'$, then $\X$ is rejection-proof and the algorithm can return YES. If, after doing this for all subsets $\X \subseteq \S$ of size at most $k\cdot d$, no rejection-proof packing $\X$ is found, the algorithm correctly returns NO.

    Next, we analyze the running time of this brute-force algorithm. For ease of reading, we initially describe this running time in terms of $n_2$, the number of elements in $I_2$, and $m_2$, the number or sets in $I_2$. Only at the end do we substitute $m_2 = \Oh\left( (k^{d+1})^d \right)$. Moreover, since we assumed that every element in $I_2$ appears in at least one set of $\S$ and every set contains at most $d$ elements, we get that $n_2 \leq d \cdot m_2$.
    
    The algorithm iterates over all subsets $\X \subseteq \S$ of size at most $k\cdot d$. There can be at most $(m_2)^{k\cdot d}$ such sets. For every such $\X$, the algorithm again iterates over all subsets $\X' \subseteq \S$ of size at most $k\cdot d$ and spends polynomial time in each iteration. This gives an upper bound on the total running time of $(m_2)^{k\cdot d} \cdot (m_2)^{k\cdot d} \cdot (n_2 + m_2)^{\Oh(1)}$. Combining and rewriting the two exponential terms while substituting that $n_2 \leq d\cdot m_2$, we obtain $2^{2k\cdot d \log(m_2)} \cdot ((d+1)m_2)^{\Oh(1)}$. Finally, treating $d$ as a constant and substituting $m_2 = \Oh\left( (k^{d+1})^d \right)$ so that $\log(m_2) \leq \Oh(k)$, we obtain a running time of $2^{\Oh(k \log k)} \cdot k^{\Oh(1)} = 2^{\Oh(k \log k)}$.

    Hence, running the kernel (which takes polynomial time in the original input size) and the brute-force algorithm in sequence takes $2^{\Oh(k \log k)} + (n+m)^{\Oh(1)}$ time in total.
\end{proof}
\section{Omitted proofs from Section 4}
\ReductionSound*
\begin{claimproof}
        Let $\varphi: V(H) \rightarrow V(G)$ be a subgraph isomorphism from $H$ to a subgraph of $G$. We show how the function $\varphi$ can be used to construct a packing $\C$ of $3$-cycles that covers all vertices of $G'$ except $z$. As such, it covers $k$ vertices. \autoref{fig:eth-reduction-isomorphisms} shows an example on the left.
        
        First, for every pair $u \in V(H)$ and $x \in V(G)$ such that $x = \varphi(u)$, we add the I-cycle $\langle u_1, u_2, x_3 \rangle$ to our solution $\C$. The union of all such cycles covers all vertices of $W_1$ and $W_2$ and it covers $n_H$ of the vertices in $W_3$.

        We use F-cycles to cover the remaining $n_G - n_H$ vertices of $W_3$. Note that for every combination of a vertex $x_3 \in W_3$ and an integer $i \in [n_G -n_H]$, the graph $G'$ contains the F-cycle $\langle v_3, i_4, i_5 \rangle$. Hence, covering the remaining vertices of $W_3$ is simple: match every uncovered vertex $x_3 \in W_3$ to an arbitrary but unique integer $i \in [n_G -n_H]$ and add the corresponding F-cycle $\langle v_3, i_4, i_5 \rangle$ to $\C$. 
        
        After the previous step, $\C$ also covers all vertices from $W_4$ and $W_5$. As such, it covers all vertices from $G'$ except $z$. Recall that $G'$ has $3\cdot n_G + 1$ vertices, which means that $\C$ covers $k = 3 \cdot n_G$ vertices. It remains to show that $\C$ is rejection-proof.

        Suppose, for contradiction, that $\C$ is not rejection-proof. Then, one of the two agents can propose an alternative packing that covers more vertices of $V_1'$ or $V_2'$ respectively than $\C$ does. Since $\C$ covers all vertices from $G'$ except $z \in V_1'$, only agent $1$ can propose an alternative packing that covers more vertices of its corresponding vertex set. Moreover, to cover more vertices from $V_1'$ than $\C$, the alternative packing must cover all vertices from $V_1'$. We use this insight to derive two more facts about rejections in this setting. 

        First, we argue that agent $1$ does not reject F-cycles. After all, if they did, this would leave the vertices in the rejected F-cycle uncovered, including in particular some vertex $i_5 \in W_5 \subseteq V_1'$. Since the vertices in $W_5$ do not appear in any internal cycles, there are no cycles that agent $1$ could propose to add in order to cover vertex $i_5$. We record this statement for later use as the following fact.

        \begin{fact} \label{obs:no-reject-f-cycles}
            Agent $1$ does not reject F-cycles.
        \end{fact}

        For our second fact, we additionally note that, to cover all vertices in $V_1'$, agent $1$ must in particular propose to add some cycle $C_z$ that covers vertex $z$. Now, we argue that agent $1$ does not reject I-cycles that are vertex-disjoint from $C_z$.

        Suppose, for contradiction, that agent $1$ does reject such a cycle and consider the vertex $x_3 \in W_3 \subseteq V_1'$ that this I-cycle contains. By rejecting this cycle, $x_3$ is left uncovered since it is also not covered by $C_z$ by assumption. The only internal cycles that $x_3$ is part of are R-cycles and N-cycles, all of which also contain $z$. Hence, agent $1$ cannot propose to add any of these to the given packing, as none of these cycles are disjoint from the cycle $C_z$ that they have already proposed to add. As such, $x_3$ cannot be covered in an alternative packing proposed by agent $1$. This contradicts the observation that such an alternative packing must cover all vertices from $V_1'$. We record the obtained statement as the following fact for later~use.

        \begin{fact} \label{obs:no-reject-i-cycles}
            Agent $1$ does not reject I-cycles that are vertex-disjoint from $C_z$.
        \end{fact}

        Now, we use the two facts above to derive a contradiction. We do so in a case distinction on the type of cycle that $C_z$ is. As $C_z$ is an internal cycle that covers $z$, it is either an R-cycle or an N-cycle.        
        
        \subparagraph{Case 1: suppose that $C_z$ is an R-cycle.} Let $C_z = \langle x_3, z \rangle$ for some $x_3 \in W_3$. As agent~$1$ was said to modify $\C$ by including this cycle, the cycle~$C$ from~$\C$ that covers~$x_3$ must be rejected to maintain a set of pairwise disjoint cycles. By \autoref{obs:no-reject-f-cycles},~$C$ is not an F-cycle. Since~$\C$ only contains I-cycles and F-cycles, it follows that~$C$ is an I-cycle. 
        
        Let $C$ be the I-cycle $\langle u_1, u_2, x_3 \rangle$. Rejecting this cycle leaves the vertex $u_1 \in V_1'$ uncovered, implying that it must be covered by some internal cycle. Vertices from $W_1$ only appear in internal cycles with other vertices from $W_1$. Including any such internal cycle $C'$ would require agent $1$ to reject the cycle(s) from $\C$ that already cover a vertex from $C'$. Since all vertices of $W_1$ are covered by an I-cycle in $\C$, agent $1$ has to reject at least one more such cycle. This contradicts \autoref{obs:no-reject-i-cycles}, since agent $1$ would have to reject an I-cycle different from $C$, the only one that intersects $C_z$.

        \subparagraph{Case 2: suppose that $C_z$ is an N-cycle.} Let $C_z = \langle x_3,y_3,z \rangle$ for some pair of distinct vertices $x_3,y_3 \in W_3$. By including this cycle, the cycles $C_x \in \C$ and $C_y \in \C$ that cover $x_3$ and $y_3$ respectively must be rejected. $\C$ only contains I-cycles and F-cycles, but we see that neither $C_x$ nor $C_y$ is an~F-cycle by \autoref{obs:no-reject-f-cycles}. Hence, both cycles are I-cycles. Because I-cycles contain only one vertex from $W_3$, $C_x$ and $C_y$ are in particular \emph{distinct} I-cycles. 
        
        We obtain that there are distinct vertices $u,v \in V(H)$ such that $C_x = \langle u_1, u_2, x_3 \rangle$ and~$C_y = \langle v_1, v_2, y_3 \rangle$ respectively. To be able to reject $\C$, agent $1$ must add one or more cycles to cover $u_1$ and $v_1$ as these vertices belong to $V_1'$. Additionally, by \autoref{obs:no-reject-i-cycles}, agent~$1$ cannot reject any more I-cycles to uncover additional vertices from~$W_1$ as~$C_x$ and~$C_y$ --- the only cycles that intersect~$C_z$ --- have already been rejected. Combined with the fact that vertices from $W_1$ only appear in internal cycles with other vertices from $W_1$, we conclude that agent $1$ has to include the A-cycle $\langle u_1, v_1 \rangle$ to cover $u_1$ and $v_1$. To arrive at our final contradiction, we show that this cycle does not exist in $G'$.
        
        By construction of the packing $\C$, the fact that $C_x = \langle u_1, u_2, x_3 \rangle$ and $C_y = \langle v_1, v_2, y_3 \rangle$ are in this packing means that the subgraph isomorphism $\varphi$ maps $u$ to $x$ and $v$ to $y$. Additionally, by step \ref{itm:N-edge} in the construction of the edge set of $G'$, the existence of the cycle $\langle x_3,y_3,z \rangle$ implies that $\{x, y\} \notin E(G)$. Now, because $\varphi$ is a subgraph isomorphism, we know that $\{u,v\} \notin E(H)$. In turn, this implies that neither $(u_1, v_1)$ nor $(v_1, u_1)$ exists in $G'$. Hence, the A-cycle between these two vertices also does not exist in $G'$.

        This yields the final contradiction to conclude that $\C$ is rejection-proof.
    \end{claimproof}
\section{Omitted proofs from Section 5}

\thmOneRejection*
\begin{proof}[Proof of \autoref{thm:1-rejection-algo}]
    Consider an instance $I=(U_1,\dots,U_p,\mathcal{S},k)$ of $1$-RP-$d$-SP and let~$U := \bigcup_{i \in [p]} U_i$. Recall that~$|U|=n$ and $|\mathcal{S}|=m$. The high-level approach is to reduce the $1$-RP-$d$-SP instance to a sequence of instances of \textsc{Exact Set Cover}. Each such instance consists of a set system~$\mathcal{S}'$ over a universe~$U'$ and asks whether~$U'$ can be covered by a collection of \emph{disjoint} sets from~$\mathcal{S}'$. The outline of the algorithm is as follows.
    \begin{enumerate}
        \item We create a set~$\mathcal{I}$ of \textsc{Exact Set Cover} instances of the form~$(U' \subseteq U, \mathcal{S}' \subseteq \S[U'])$ with the following properties:
        \begin{itemize}
            \item any collection of disjoint sets $\mathcal{X} \subseteq \mathcal{S}'$ that completely covers $U'$ is a solution for the $1$-RP-$d$-SP instance $I$, and
            \item for any solution~$\mathcal{X}$ to instance~$I$, there is an \textsc{Exact Set Cover} instance $(U' = \bigcup_{A \in \mathcal{X}} A, \mathcal{S}')$ in~$\mathcal{I}$ with~$\mathcal{X} \subseteq \mathcal{S}'$.
        \end{itemize}
        \item We run the algorithm for \textsc{Exact Set Cover} from \autoref{lm:exact set cover} on each pair~$(U', \mathcal{S'}) \in \mathcal{I}$, and output YES if and only if some instance in~$\mathcal{I}$ has an exact set cover.
    \end{enumerate}
    It is easy to verify that when~$\mathcal{I}$ has both stated conditions, the results of the \textsc{Exact Set Cover} computations yield the correct answer for $1$-RP-$d$-SP. Hence the construction of~$\mathcal{I}$ reduces the task of solving $1$-RP-$d$-SP to a sequence of \textsc{Exact Set Cover} instances whose universe has size at most~$|U|$. The set~$\mathcal{I}$ is constructed by the following process.
    \begin{itemize}
        \item For each set~$U' \subseteq U$ of size at least~$k$, for which there is no agent~$i \in [p]$ that has an $i$-internal set disjoint from~$U'$, do as follows. Let~$\mathcal{S}'$ contain all sets~$A \in \mathcal{S}[U']$ satisfying
        \begin{itemize}
            \item for each agent~$i \in [p]$, each collection of disjoint sets~$\mathcal{R} \subseteq \mathcal{S}[U_i \setminus (U' \setminus A)]$ covers at most~$|A \cap U_i|$ elements from~$U_i$,
        \end{itemize}
        and add the pair~$(U', \mathcal{S}')$ to~$\mathcal{I}$.
    \end{itemize}
    Intuitively, we consider all sets~$U' \subseteq U$ that could be the set of elements covered by a solution (hence~$|U'| \geq k$). Then we consider which sets~$A \in \mathcal{S}$ could be used by a $1$-rejection-proof solution that covers exactly~$U'$. All relevant sets are contained in~$\mathcal{S}[U']$ as otherwise the solution covers an element outside~$U'$. We remove any set~$A$ from consideration if it is clear that using the set~$A$ in a candidate solution that covers~$U'$ would lead to a rejection.

    Next, we prove that the constructed set~$\mathcal{I}$ has both claimed properties. We start by proving that for any pair $(U',\mathcal{S'}) \in \mathcal{I}$, any collection of disjoint sets $\mathcal{X} \subseteq \mathcal{S}'$ covering all elements of $U'$ is a solution for $I$. By construction, the set $U'$ contains at least $k$ elements and $\mathcal{X}$ is a collection of disjoint sets. It remains to show that $\mathcal{X}$ is 1-rejection-proof. Assume, for sake of contradiction, that some agent $i$ rejects sets $\X_{\operatorname{rej}}$ and replaces these sets with the sets $\X_{\operatorname{int}}$. If $\X_{\operatorname{rej}}$ is empty this would imply that there exists an $i$-internal set in~$\mathcal{S}$ that is disjoint from $U'$; but then the pair $(U',\mathcal{S}')$ would not have been added to~$\mathcal{I}$. Hence~$\X_{\operatorname{rej}}$ is non-empty, and since the problem $1$-RP-$d$-SP allows rejection of at most one set, we can consider the unique set $R \in \mathcal{X}_{\operatorname{rej}}$ that is rejected. By definition of rejection, the sets in $\mathcal{X}_{\operatorname{int}}$ cover more elements of agent $i$ than $R$ does and $\mathcal{X}_{\operatorname{int}} \subseteq \mathcal{S}[U_i \setminus (U' \setminus R)]$. But this shows that the condition for including~$A$ in~$\mathcal{S}'$ is violated. Hence any collection of disjoint sets from~$\mathcal{S}'$ covering~$U'$ is indeed a solution to the $1$-RP-$d$-SP instance~$I$.

    We now prove that~$\mathcal{I}$ has the second property. So suppose there exists a solution $\mathcal{X}$ for $I$. We will prove that there is a pair $(U', \mathcal{S}')$ in~$\mathcal{I}$ such that $\mathcal{X} \subseteq \mathcal{S}'$ and $U' =\bigcup_{A \in \mathcal{X}} A$. Since $\X$ is a solution for $I$, there cannot be an internal set disjoint from $U'$ (it would lead to a rejection) and $|U'| \ge k$. Hence a pair must be created with $U'$. If some set $A \in \mathcal{X}$ does not meet the requirement for being included in the corresponding collection~$\mathcal{S}'$, then the agent for which~$A$ is an internal set rejects~$\X$ since~$A$ can be replaced by a family of disjoint sets in~$\mathcal{S}[U_i \setminus (U' \setminus A)]$ (which are therefore also disjoint from the remaining sets in~$\X$) to cover more elements from~$U_i$. This would contradict that $\X$ is a solution. Hence the collection~$\mathcal{S}'$ constructed for~$U'$ has~$\X \subseteq \mathcal{S}'$ and~$(U', \mathcal{S}') \in \mathcal{I}$. This establishes the second property for~$\mathcal{I}$ and therefore proves correctness of the algorithm.
    
    It remains to bound the running time. There are $\sum_{i=k}^n \binom{n}{i}$ subsets of $U$ of size at least $k$. Checking for a subset $U'$ if there is an internal set disjoint from $U'$ and calculating $\mathcal{S}[U']$ takes polynomial time in $n+m$. For each choice of $U'$, we check for at most $m$ sets for at most $p$ agents if keeping this set could lead to rejections. The evaluation is only done when for every agent $i$ it holds that $\mathcal{S}[U_i \setminus U'] = \emptyset$. Then for any set $A$ for which the condition is evaluated, a pairwise collection of disjoint sets $\mathcal{R} \subseteq \mathcal{S}[U_i \setminus (U ' \setminus A)]$ contains at most $d$ sets since each such set contains an element of~$A$ which has size at most~$d$. Therefore, this check can be done in time $O(n \cdot m^d)$. 
    As~$d \in \Oh(1)$, in total Step 1 takes $\sum_{i=k}^n \binom{n}{i} \cdot(n + m )^{\Oh(1)}$ time.  
    Step 2 takes at most $\sum_{i=0}^n \binom{n}{i} \cdot 2^{i} \cdot 1^{n-i} \cdot (m+n)^{\Oh(1)} = 3^n \cdot (m+n)^{\Oh(1)}$ time, which follows from the binomial theorem and \autoref{lm:exact set cover}. This concludes the proof.
\end{proof}

\end{document}